\pdfoutput=1
\documentclass[sigconf]{acmart}
\usepackage{framed}
\usepackage{hyperref}
\usepackage{balance}
\usepackage{natbib}
\usepackage{epsfig}
\usepackage{color}
\usepackage{subfigure}
\usepackage{multirow,tabularx}

\usepackage{mathtools}
\usepackage{graphicx}
\usepackage{epstopdf}
\usepackage{bm}

\usepackage{csquotes}
\usepackage{array}
\usepackage[yyyymmdd,hhmmss]{datetime}
\usepackage[subject={Todo}]{pdfcomment}
\usepackage[textsize=scriptsize,bordercolor=black!20]{todonotes}
\usepackage{xcolor,colortbl}
\usepackage{pifont}% http://ctan.org/pkg/pifont
\usepackage{lipsum,environ,amsmath}
\usepackage{slashbox,booktabs,amsmath}
\usepackage[ruled,vlined,linesnumbered]{algorithm2e}

\usepackage{xspace}

\usepackage{subcaption}

\usepackage{siunitx}

\newcounter{Lcount}
\newcommand{\numsquishlist}{
   \begin{list}{\arabic{Lcount}. }
    { \usecounter{Lcount}
 \setlength{\itemsep}{-.1ex}      \setlength{\parsep}{0ex}
      \setlength{\topsep}{0ex}       \setlength{\partopsep}{0ex}
      \setlength{\leftmargin}{1em} \setlength{\labelwidth}{1em}
      \setlength{\labelsep}{0.1em} } }
\newcommand{\numsquishend}{\end{list}}

\newcommand{\squishlist}{
   \begin{list}{$\bullet$}
    { \setlength{\itemsep}{-.1ex}      \setlength{\parsep}{0ex}
      \setlength{\topsep}{0ex}       \setlength{\partopsep}{0ex}
      \setlength{\leftmargin}{.8em} \setlength{\labelwidth}{1em}
      \setlength{\labelsep}{0.5em} } }
\newcommand{\squishend}{\end{list}}

\definecolor{HLcolor}{rgb}{0,0.0,0}
\AtBeginDocument{%
  \providecommand\BibTeX{{%
    \normalfont B\kern-0.5em{\scshape i\kern-0.25em b}\kern-0.8em\TeX}}}

\begin{document}

\title{LGTD: Local–Global Trend Decomposition for Season-Length–Free Time Series Analysis}


\author{Chotanansub Sophaken}
\orcid{0009-0004-5165-006X}
\affiliation{%
  \institution{King Mongkut’s University of Technology Thonburi}
  \country{Thailand}}
% \affiliation{%
%   \institution{Securities and Exchange Commission}
%   \streetaddress{333 3 Vibhavadi Rangsit Rd, Chom Phon, Chatuchak}
%   \city{Bangkok}
%   \country{Thailand}
%   \postcode{10900}
% }
\email{chotanansub.s@gmail.com}

\author{Thanadej Rattanakornphan}
\affiliation{%
  \institution{Department of Computer Engineering, Kasetsart University}
  \streetaddress{}
  \city{Bangkok}
  \country{Thailand}
}
\email{ra.thanadej@gmail.com}

\author{Piyanon Charoenpoonpanich}
\affiliation{%
  \institution{Independence}
  \streetaddress{}
  \city{Bangkok}
  \country{Thailand}
}
\email{piyanon.charoenpoonpanich@gmail.com}

\author{Thanapol Phungtua-eng}
\affiliation{%
  \institution{Rajamangala University of Technology Tawan-ok
}
  \country{Thailand}}
% \affiliation{%
%   \institution{Securities and Exchange Commission}
%   \streetaddress{333 3 Vibhavadi Rangsit Rd, Chom Phon, Chatuchak}
%   \city{Bangkok}
%   \country{Thailand}
%   \postcode{10900}
% }
\email{thanapol_ph@rmutto.ac.th}

\author{Chainarong Amornbunchornvej}
\orcid{0000-0003-3131-0370}
\affiliation{%
  \institution{National Electronics and Computer Technology Center}
  \streetaddress{112 Phahonyothin Road, Khlong Nueng, Khlong Luang District}
  \city{Pathumthani}
  \country{Thailand}
  \postcode{12120}}
\email{chainarong.amo@nectec.or.th}

\renewcommand{\shortauthors}{C. Sophaken, T Rattanakornphan, P. Charoenpoonpanich, T. Phungtua-eng, and C. Amornbunchornvej}

\begin{abstract}
Time series decomposition into trend, seasonal, and residual
components is a fundamental primitive in data management and
analytics pipelines, underpinning anomaly detection, change-point
analysis, and forecasting. Most existing methods require a
user-specified or estimated season length and implicitly assume
stable periodic structure. In large, heterogeneous collections---where
recurring patterns drift, appear intermittently, or operate at
multiple nonstationary scales---period selection becomes brittle and
per-series tuning does not scale.
We propose \textsc{LGTD} (Local--Global Trend Decomposition), a
season-length--free decomposition framework that requires no period
specification and operates with a single fixed default configuration
across datasets. \textsc{LGTD} represents a series as the sum of
(i)~a smooth global trend capturing long-term evolution, (ii)~adaptive
local trends inferred by an error-driven local linear segmentation
procedure, and (iii)~a residual component. Rather than modeling
seasonality through an explicit periodic basis, \textsc{LGTD} treats
it as an emergent property arising from the recurrence of local trend
regimes, decoupling decomposition quality from any predefined or
estimated season length.
We prove that the local trend inference procedure terminates in a
bounded number of refinement iterations and runs in linear time in the
series length, independent of any seasonal parameter, and confirm this
empirically: \textsc{LGTD} scales linearly in both runtime and memory
and is the fastest method across all tested lengths, while several
baselines degrade super-linearly. On synthetic benchmarks \textsc{LGTD}
achieves balanced decomposition accuracy across fixed, transitive, and
variable season-length regimes, particularly where period-based methods
degrade, and on real-world datasets it yields interpretable components
and low-structure residuals under irregular temporal dynamics---all
using a single default configuration. The complete source code and
datasets used in this work are publicly available at
\url{https://github.com/chotanansub/LGTD}, ensuring full
reproducibility of the reported results.
\end{abstract}

% \url{https://anonymous.4open.science/r/LGTD-69}




\received{20 December 2024}
\received[revised]{11 April 2025}
\received[accepted]{9 June 2025}

\maketitle

\section{Introduction}

Time-series decomposition into trend, seasonal structure, and residual components
is a foundational operation in knowledge discovery, supporting tasks such as
anomaly detection, change-point analysis, forecasting, and representation
learning. Classical methods such as STL~\cite{cleveland1990stl} and
regression-based variants~\cite{dokumentov2022str} remain widely used because they
produce interpretable components and integrate naturally into downstream
pipelines. However, these approaches rely on a strong structural assumption: the
existence of a user-specified or estimable \emph{season length}, with seasonality
modeled as a stable, approximately periodic signal. This assumption is frequently
violated in practice, particularly in heterogeneous collections where recurring
patterns drift, appear intermittently, or operate at multiple, nonstationary
time scales.

A closely related challenge concerns the representation of trend structure.
Real-world time series often exhibit regime changes, abrupt transitions, and
event-driven dynamics that are poorly captured by a single smooth trend.
Optimization-based estimators such as $\ell_1$ trend filtering~\cite{10.1137/070690274}
address this issue by modeling trends as piecewise-linear signals, interpreting
changes in slope as structural events. This perspective suggests that temporal
structure is often better understood as a sequence of local regimes rather than
as a globally smooth function, especially in operational and sensor-driven data.

Recent work has improved the robustness and scalability of seasonal--trend
decomposition, yet most methods remain fundamentally tied to explicit
season-length choices. MSTL extends STL-style decomposition to multiple seasonal
periods~\cite{bandara2021mstlseasonaltrenddecompositionalgorithm}. RobustSTL and
Fast RobustSTL improve robustness to anomalies and complex patterns
\cite{10.1609/aaai.v33i01.33015409,10.1145/3394486.3403271}. Online and
database-oriented variants further optimize scalability and streaming behavior
\cite{10.14778/3523210.3523219,10.14778/3583140.3583155,10.1145/3637528.3671510,11112870}.
Despite these advances, existing frameworks continue to define seasonality
through explicit periodic templates—either fixed or estimated—making
decomposition quality sensitive to period selection and limiting robustness
under irregular or drifting temporal structure.

\begin{figure*}[t]
    \centering
    \includegraphics[width=1\linewidth]{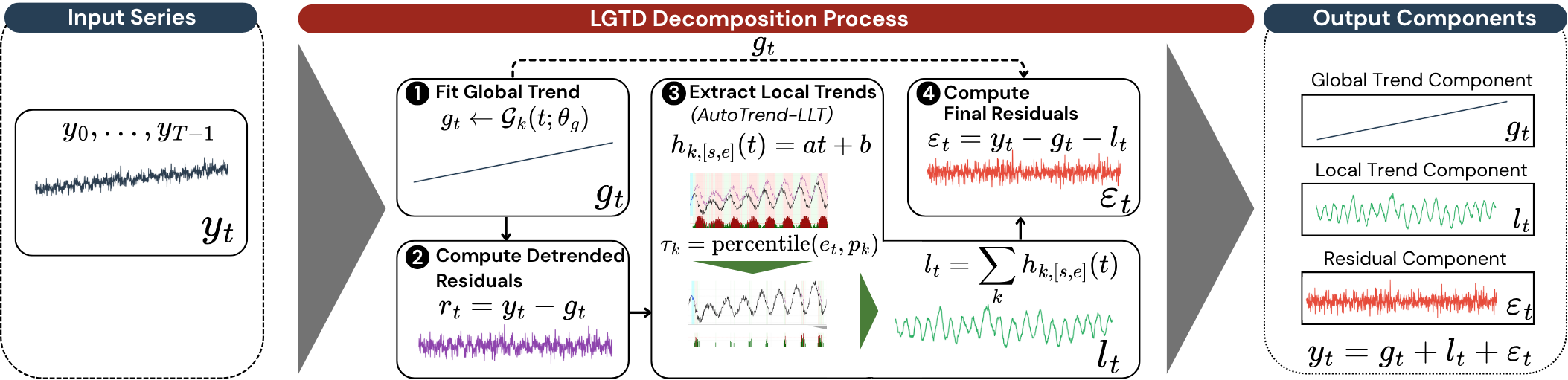}
    \caption{\emph{LGTD global--local decomposition process. After estimating a global trend $g_t$, the detrended series is segmented adaptively by AutoTrend-LLT to infer local trends $\ell_t$, whose recurrent structure gives rise to emergent seasonality without requiring a predefined season length. Removing both components yields residuals $\varepsilon_t$, giving $y_t = g_t + \ell_t + \varepsilon_t$.}}
    \label{fig:LGTD_pipeline}
\end{figure*}

This paper introduces \textsc{LGTD} (Local--Global Trend Decomposition), a
season-length--free decomposition framework that redefines seasonality as an
\emph{emergent structural property}. LGTD represents a time series as
\[
    y_t = g_t + \ell_t + \varepsilon_t,
\]
where $g_t$ captures long-term evolution via a smooth global trend, $\ell_t$
captures short-lived, piecewise-linear local regimes, and $\varepsilon_t$
contains residual variation. In LGTD, seasonality does not correspond to an explicit periodic basis or fixed-frequency signal; instead, it is defined operationally as the recurrence of similar local trend regimes after removal of the global trend. The seasonal component is thus the reconstructed local-trend signal whose repeated structural patterns give rise to implicit seasonality even when timing is irregular or drifting. While the distinction between trend and seasonality is not always uniquely identifiable in real-world time series, LGTD adopts a structural interpretation in which long-term evolution is assigned to the global trend and recurring local patterns are treated as manifestations of seasonality

Concretely, LGTD first estimates a global trend, then applies an adaptive,
error-driven local linear trend inference procedure (AutoTrend-LLT) to segment
the detrended signal into local regimes (Fig.~\ref{fig:LGTD_pipeline}). By
decoupling decomposition from any season-length specification, LGTD provides a
structural alternative to period-based methods for analyzing nonstationary time
series.

Our contributions are:
\begin{itemize}
    \item \textbf{Season-length--free decomposition.}
    We introduce a structural formulation in which seasonality emerges from the
    recurrence of data-driven local trend regimes rather than from explicit
    periodic modeling.
    \item \textbf{Adaptive local trend inference.}
    We propose AutoTrend-LLT, an error-driven segmentation procedure that
    identifies piecewise-linear local regimes while adapting to structural
    breaks and evolving dynamics.
    \item \textbf{Theoretical guarantees.}
    We establish a finite-termination bound for the local trend inference
    procedure and show that LGTD admits linear-time complexity in series length,
    independent of any seasonal parameter.
\end{itemize}

\section{Related Work}

\paragraph{Classical seasonal--trend decomposition.}
STL~\cite{cleveland1990stl} is a foundational seasonal--trend decomposition method based on repeated LOESS smoothing and requires a user-specified seasonal period, making it suitable primarily for single-seasonal time series. To address multiple seasonalities, MSTL extends STL by sequentially extracting several seasonal components~\cite{bandara2021mstlseasonaltrenddecompositionalgorithm}. While effective when relevant periods are known, multi-seasonal methods remain highly sensitive to season-length selection. In practice, heuristics such as using a large seasonal window are often employed, but they still assume prior knowledge of a period close to the ground truth, limiting automation~\cite{10.1609/aaai.v33i01.33015409,10.14778/3583140.3583155}. 

STR~\cite{dokumentov2022str} reformulates decomposition as a regression problem, preserving interpretability while still assuming explicit seasonal structure. Overall, classical methods perform well under stable and known periodicity, but their reliance on fixed seasonal lengths hinders deployment across heterogeneous collections.

\paragraph{Robust decomposition under anomalies and complex patterns.}
RobustSTL improves resilience to outliers, seasonality shifts, and abrupt changes through robust trend estimation~\cite{10.1609/aaai.v33i01.33015409}, while FastRobustSTL extends this framework to multiple seasonalities with improved efficiency~\cite{10.1145/3394486.3403271}. Both are grounded in $\ell_1$ trend filtering~\cite{10.1137/070690274}, which encourages piecewise-linear trends and yields interpretable structural breakpoints. This regime-based view of trend is well suited to event-driven time series. However, these methods require careful hyperparameter tuning and remain dependent on predefined seasonal periods, limiting their suitability for fully automated settings.

\paragraph{Online and systems-oriented decomposition.}
Several recent works focus on scalable or low-latency decomposition in streaming environments. OnlineSTL accelerates online decomposition for long seasonalities~\cite{10.14778/3523210.3523219}, while OneShotSTL and BacktrackSTL achieve $O(1)$ update time with different robustness and design trade-offs~\cite{10.14778/3583140.3583155,10.1145/3637528.3671510}. In database contexts, OneRoundSTL adapts decomposition to LSM-tree storage and missing data via page-level precomputation~\cite{11112870}. Although these systems emphasize scalability and deployment constraints, most still encode or assume an explicit seasonal period.

\paragraph{Season-length estimation under nonstationarity.}
When seasonal length is unknown, a common pipeline estimates dominant periods before applying decomposition. Representative methods include AutoPeriod~\cite{10.1137/1.9781611972757.40}, CFD-AutoPeriod~\cite{10.1007/978-3-030-39098-3_4}, and SAZED~\cite{Toller2019}. To handle evolving periodicity, adaptive approaches such as ASTD integrate online season-length estimation into the decomposition process~\cite{10.1007/978-3-031-70344-7_25}. These methods highlight a key challenge in practice: identifying an appropriate period can be as difficult as decomposition itself when cycles drift or appear intermittently. Despite adaptivity, they still rely on explicit season-length parameters.
\section{Methods}

This section describes the \textsc{LGTD} (Local--Global Trend Decomposition)
framework. Fig.~\ref{fig:LGTD_pipeline} summarizes the overall pipeline.
Given an input series $y_t$, \textsc{LGTD} first estimates a smooth global trend
$g_t$ to capture long-term structure. The detrended residuals are then processed
by the AutoTrend-LLT procedure to extract piecewise-linear local trends
$\ell_t$, whose recurring patterns act as emergent seasonality without
requiring a predefined season length. Subtracting both components yields the
final residuals $\varepsilon_t$, giving the decomposition
$y_t = g_t + \ell_t + \varepsilon_t$. The following subsections detail each
stage of the method.

\subsection{LGTD: Unified Global--Local Trend Decomposition}
\label{subsec:lgtd_full}

Traditional time--series decomposition methods require explicit specification
of a season length and assume that seasonal structure is stable and periodic.
These assumptions limit applicability in heterogeneous collections where cycles
may drift, appear intermittently, or occur at multiple scales. In contrast,
\textsc{LGTD} reframes seasonality as an \emph{emergent} phenomenon arising from
the repetition of data-driven local trend regimes. This eliminates the need to
specify a season length and allows the decomposition to adapt automatically to
irregular or evolving temporal patterns.

\paragraph{Model Structure.}
Given a univariate series $y_0,\dots,y_{T-1}$, \textsc{LGTD} decomposes it as

\begin{equation}
    y_t = g_t + \ell_t + \varepsilon_t,
\end{equation}

where $g_t$ is a smooth global trend capturing long-term evolution, $\ell_t$ is
a piecewise-linear local trend capturing short-term regimes and emergent
seasonality, and $\varepsilon_t$ contains residual fluctuations. The global
trend is obtained using any regression or smoothing model $\mathcal{G}$ (e.g.,
linear regression, spline smoothing, long-window local regression). Local
trends are inferred by applying AutoTrend-LLT
(Algorithm~\ref{algo:autotrend_llt}) to the global-trend residuals.

\paragraph{Global Trend Extraction.}
The series is first smoothed using $\mathcal{G}$ with hyperparameters
$\theta_g$, producing estimates $g_t = \mathcal{G}(t;\theta_g)$. Detrending
yields residuals $r_t = y_t - g_t$, which serve as input to AutoTrend-LLT.

\paragraph{Local Trend Extraction (Emergent Seasonality).}
Local trends are extracted from the detrended residuals using AutoTrend-LLT
(Section~\ref{subsec:local_trend}). In the LGTD decomposition, the local-trend
component is identified with the LLT reconstruction, i.e.,
$\ell_t := \hat r_t$, so that recurring local-trend regimes act as emergent
seasonality even when timing is irregular or drifting.

\paragraph{Final Residuals.}
Subtracting both global and local trends yields the final residual component
$\varepsilon_t = y_t - g_t - \ell_t$.

\paragraph{Algorithm.}
Algorithm~\ref{algo:LGTD} summarizes the complete \textsc{LGTD}
decomposition procedure.

\begin{algorithm}[t]
\caption{LGTD — Local–Global Trend Decomposition}
\label{algo:LGTD}
\DontPrintSemicolon
\small

\KwIn{Series $y_0,\dots,y_{T-1}$; global-trend model $\mathcal{G}$ with hyperparameters $\theta_g$;\\
      local-trend parameters $(w, K_{\max}, p_0, \Delta p,$ \texttt{update\_threshold}).}
\KwOut{Global trend $\mathbf{g}$; local trend $\boldsymbol{\ell}$; residuals $\boldsymbol{\varepsilon}$;\\
       local-trend labels $\mathbf{m}$; local models $\mathcal{M}$.}

\BlankLine
\textbf{Step 1: Fit global trend.}\;
Fit $\mathcal{G}$ on $\{(t,y_t)\}_{t=0}^{T-1}$ and compute
$g_t \leftarrow \mathcal{G}(t;\theta_g)$.\;
Set $\mathbf{g} \leftarrow (g_0,\dots,g_{T-1})$.\;

\BlankLine
\textbf{Step 2: Compute detrended residuals.}\;
$r_t \leftarrow y_t - g_t$ for all $t$; set $\mathbf{r} \leftarrow (r_0,\dots,r_{T-1})$.\;

\BlankLine
\textbf{Step 3: Extract local trends via LLT.}\;
$(\mathbf{m}, \hat{\mathbf{r}}, \mathcal{M}) \leftarrow
\textsc{LocalLinearTrend}(\mathbf{r}, w, K_{\max}, p_0, \Delta p, \texttt{update\_threshold})$;\;
set $\boldsymbol{\ell} \leftarrow \hat{\mathbf{r}}$, where $\ell_t := \hat r_t$.\;

\BlankLine
\textbf{Step 4: Compute final residuals.}\;
$\varepsilon_t \leftarrow y_t - g_t - \ell_t$ for all $t$; set
$\boldsymbol{\varepsilon} \leftarrow (\varepsilon_0,\dots,\varepsilon_{T-1})$.\;

\BlankLine
\Return $\mathbf{g}, \boldsymbol{\ell}, \boldsymbol{\varepsilon}, \mathbf{m}, \mathcal{M}$.\;

\end{algorithm}

\subsection{Adaptive Local Linear Trend Decomposition (AutoTrend-LLT)}

\label{subsec:local_trend}

A core component of \textsc{LGTD} is the identification of \emph{local trends}—
short-lived linear regimes whose repetition or evolution gives rise to emergent
seasonal behavior. Rather than assuming a fixed season length or a predefined
number of components, LGTD discovers local trends directly from data using an
iterative, error-driven segmentation procedure. This design enables LGTD to
capture irregular, drifting, or multi-scale patterns while remaining sensitive
to structural breaks.

\paragraph{Local Trend Modeling.}
Given a univariate series $y_0,\dots,y_{T-1}$, the method maintains a dynamic set
of \emph{focus targets}—indices whose local trend assignments remain uncertain.
At each iteration, these indices are grouped into disjoint contiguous
\emph{focus ranges} $\{[s_j, e_j]\}$, where each range denotes an inclusive
integer index interval. Let $w$ denote the local fitting window size. Each focus
range receives \emph{exactly one} local model in that iteration: an ordinary
least squares (OLS) linear regression fitted on the $w$ observations
immediately preceding the range. Specifically, for a focus range $[s,e]$, the
model is fitted on observations indexed by
$t_{\mathrm{start}}=\max(0,s-w)$ through
$t_{\mathrm{end}}=\min(e-1,t_{\mathrm{start}}+w-1)$.
This window provides the slope and intercept estimated by OLS, which are then
extrapolated forward across the entire range.

For a focus range $[s,e]$ at iteration $k$, the fitted local model takes the
form
\begin{equation}
    h_{k,[s,e]}(t) = a_{k,[s,e]} t + b_{k,[s,e]},
\end{equation}
yielding predictions
\begin{equation}
    \tilde{y}_t = h_{k,[s,e]}(t), \qquad t \in [s,e].
\end{equation}
For consistency with Algorithm~\ref{algo:autotrend_llt}, we denote the
per-index prediction array by $\tilde{\mathbf{y}}$ with
$\tilde{\mathbf{y}}[t] := \tilde{y}_t$. Absolute prediction errors are then
computed as
\begin{equation}
    e_t = |y_t - \tilde{y}_t| = |y_t - \tilde{\mathbf{y}}[t]|.
\end{equation}

\paragraph{Error-Driven Refinement.}
Prediction errors are compared against a percentile threshold $\tau$, initially
set to a baseline percentile $p_0$. Points with $e_t \le \tau$ are deemed
well-explained by the current local model and are permanently assigned to the
current iteration, while the remaining points stay in the focus set and are
reconsidered in subsequent iterations. When threshold updating is enabled, the
percentile level is increased by a fixed increment $\Delta p$ after each
iteration, thereby relaxing the acceptance criterion and promoting progressive
merging of stable regions. Iteration continues until all points are assigned or
a maximum number of iterations is reached.

This procedure yields a piecewise-linear local trend signal in which:
(i) smooth segments form broad regimes,
(ii) abrupt changes appear as short segments or boundary splits, and
(iii) repeated local patterns—even with drifting or irregular timing—naturally
function as emergent seasonal behavior.

\paragraph{Output.}
The algorithm returns (i) trend labels
$\mathbf{m} \in \{-1,1,\dots,K_{\max}\}^T$ indicating the iteration at which each
point is assigned (with $-1$ denoting unassigned points), (ii) the reconstructed
local-trend values $\hat{\mathbf{y}}$, and (iii) the collection of fitted local
trend functions. These outputs form the local-trend component of the full \textsc{LGTD}
decomposition.

\paragraph{Algorithm.}
Algorithm~\ref{algo:autotrend_llt} summarizes the adaptive local linear trend
extraction procedure described above.

\begin{algorithm}[t]
\caption{AutoTrend-LLT: Adaptive Local Linear Trend Decomposition}
\label{algo:autotrend_llt}
\DontPrintSemicolon

\KwIn{Series $y_0,\dots,y_{T-1}$; window size $w$; max iterations $K_{\max}$;\\
      baseline percentile $p_0$; step $\Delta p$; flag \texttt{update\_threshold}.}
\KwOut{Trend labels $\mathbf{m}\in\{-1,1,\dots,K_{\max}\}^T$; local trend $\hat{\mathbf{y}}$; models $\mathcal{M}$.}

\BlankLine
\textbf{Initialize:}
$\mathbf{m} \leftarrow -1$; $\hat{\mathbf{y}} \leftarrow \text{NaN}$;
$\mathcal{M} \leftarrow [\,]$;\;
$\mathcal{F} \leftarrow \{w, w+1, \dots, T-1\}$; $p \leftarrow p_0$.\;

\BlankLine
\For{$k \leftarrow 1$ \KwTo $K_{\max}$}{
  \If{$\mathcal{F} = \emptyset$}{\textbf{break}}

  Partition $\mathcal{F}$ into disjoint contiguous ranges
  $\mathcal{R}=\{[s_j,e_j]\}$.\;
  $E \leftarrow [\,]$;\;
  Initialize $\tilde{\mathbf{y}}$ as an array of length $T$ with $\text{NaN}$.\;

  \ForEach{$[s,e] \in \mathcal{R}$}{
    $t_{\text{start}} \leftarrow \max(0,\, s - w)$,\quad
    $t_{\text{end}} \leftarrow \min(e-1,\, t_{\text{start}} + w - 1)$.\;

    Fit linear regression $h_{k,[s,e]}(t) = a_{k,[s,e]} t + b_{k,[s,e]}$ on
    $\{(t,y_t) : t_{\text{start}} \le t \le t_{\text{end}}\}$.\;
    Append $h_{k,[s,e]}$ to $\mathcal{M}$.\;

    \For{$t \leftarrow s$ \KwTo $e$}{
      $\tilde{\mathbf{y}}[t] \leftarrow h_{k,[s,e]}(t)$;\;
      Append $|y_t - \tilde{\mathbf{y}}[t]|$ to $E$.\;
    }
  }

  $\tau \leftarrow \text{percentile}(E, p)$.\;
  $\mathcal{L} \leftarrow \{t \in \mathcal{F} : |y_t - \tilde{\mathbf{y}}[t]| \le \tau\}$;\;
  $\mathcal{H} \leftarrow \mathcal{F} \setminus \mathcal{L}$.\;

  \ForEach{$t \in \mathcal{L}$}{
    $\mathbf{m}[t] \leftarrow k$;\;
    $\hat{\mathbf{y}}[t] \leftarrow \tilde{\mathbf{y}}[t]$;\;
  }
  $\mathcal{F} \leftarrow \mathcal{H}$.\;

  \If{\texttt{update\_threshold}}{
    $p \leftarrow p + \Delta p$.\;
  }
}

\Return $(\mathbf{m}, \hat{\mathbf{y}}, \mathcal{M})$.\;

\end{algorithm}

\subsection{Computational Complexity}
\label{subsec:complexity}

The runtime of \textsc{LGTD} consists of (i) fitting the global trend and
(ii) extracting local trends via AutoTrend-LLT. Global trend estimation using a
regression or smoothing model $\mathcal{G}$ costs between $O(T)$ and
$O(T \log T)$ for common linear or spline-based methods.

The LLT stage performs $K^\star$ refinement iterations
(Section~\ref{subsec:termination}). In iteration $k$, one local linear
regression is fitted for each contiguous focus range, using a fixed-size window
of length $w$. Since $w$ is constant, the total cost of fitting all local models
is $O(|\mathcal{R}^{(k)}|)$, where $\mathcal{R}^{(k)}$ denotes the set of focus
ranges. Predictions and error evaluations are performed over the current focus
set $\mathcal{F}^{(k)}$, costing $O(|\mathcal{F}^{(k)}|)$. Because
$|\mathcal{R}^{(k)}| \le |\mathcal{F}^{(k)}|$ and
$|\mathcal{F}^{(k+1)}| \le |\mathcal{F}^{(k)}|$, the total LLT cost satisfies

\begin{equation}
    O\!\left(\sum_{k=1}^{K^\star} |\mathcal{F}^{(k)}|\right)
    = O(K^\star T),
\end{equation}

with $K^\star$ typically small (empirically $<10$). Thus, the overall complexity of \textsc{LGTD} is $O(K^\star T)$, i.e., linear in
the length of the time series.

\paragraph{Comparison with STL.}
Classical STL repeatedly applies local regression smoothers and relies on
multiple inner and outer iterative loops. While each pass costs $O(T)$,
convergence typically requires many iterations (often 20--50), and runtime is
tied to the user-specified season length through the configuration of the
smoothing steps. In contrast, \textsc{LGTD} has no dependence on a predefined
season length: local trends are inferred automatically, and runtime depends only
on $T$ and the small iteration count $K^\star$. As a result, \textsc{LGTD}
achieves linear-time scaling comparable to STL while removing the dominant tuning
parameter and avoiding fixed season-length assumptions.

\subsection{Finite-Termination Guarantee}
\label{subsec:termination}

We provide a finite-termination guarantee for AutoTrend-LLT, establishing an
explicit upper bound on the number of refinement iterations and supporting the
linear-time complexity of \textsc{LGTD}.

\begin{proposition}[Termination Bound for AutoTrend-LLT]
\label{prop:llt_termination}
Consider Algorithm~\ref{algo:autotrend_llt} with \texttt{update\_threshold}
enabled, baseline percentile $p_0 \in (0,100]$, and step size $\Delta p > 0$.
Assume the percentile operator is clamped so that
$\text{percentile}(E,p)=\max(E)$ for all $p \ge 100$. Define

\begin{equation}
K^\star \;:=\; 1 + \left\lceil \frac{100 - p_0}{\Delta p} \right\rceil .
\end{equation}

If $K_{\max} \ge K^\star$, then AutoTrend-LLT assigns every index
$t \in \{w,\dots,T-1\}$ a finite label by iteration $K^\star$ and terminates with
an empty focus set.
\end{proposition}

\begin{proof}
With \texttt{update\_threshold} enabled, the percentile schedule satisfies
$p^{(k)} = p_0 + (k-1)\Delta p$. By definition of $K^\star$, we have
$p^{(K^\star)} \ge 100$, so the threshold at iteration $K^\star$ equals
$\max(E)$ over $E=\{e_t : t\in\mathcal{F}^{(K^\star)}\}$. Consequently, all
remaining focused indices are classified as low-error and assigned in that
iteration, leaving an empty focus set and causing termination.
\end{proof}

This bound shows that AutoTrend-LLT always terminates in a finite number of
iterations determined solely by the percentile schedule $(p_0,\Delta p)$ and is
independent of any season length or domain-specific parameter.

\section{Experimental Evaluation}
\label{sec:expset}

We evaluate \textsc{LGTD} against seven representative state-of-the-art time
series decomposition methods on both synthetic datasets with known ground-truth
components and real-world benchmarks exhibiting non-stationary temporal
dynamics. The evaluation focuses on three key aspects: (i) decomposition
accuracy, (ii) robustness to evolving periodic structure, and (iii)
adaptability across diverse temporal patterns.

\subsection{Baseline Methods}

We compare \textsc{LGTD} with three categories of decomposition approaches.

\textbf{Classical methods.}
STL \cite{cleveland1990stl} performs seasonal--trend decomposition using
iterative LOESS smoothing under a user-specified fixed periodicity. STR
\cite{dokumentov2022str} extends STL to multiple seasonal components via
seasonal--trend regression. Both methods assume stationary seasonality and
require explicit period specification.

\textbf{Robust methods.}
FastRobustSTL \cite{10.1609/aaai.v33i01.33015409} improve robustness
to outliers using $\ell_1$ regression and non-local seasonal filtering. ASTD
\cite{10.1007/978-3-031-70344-7_25} adaptively estimates time-varying periodicity
via a sliding discrete Fourier transform, making it suitable for non-stationary
and streaming scenarios.

\textbf{Online methods.}
ASTD$_{\text{Online}}$, OnlineSTL \cite{10.14778/3523210.3523219}, and
OneShotSTL \cite{10.14778/3583140.3583155} perform incremental decomposition and
are designed for real-time or streaming data processing.

All baselines are evaluated using their recommended default configurations. For
methods requiring period specification, ground-truth periods are provided when
available to establish an upper-bound performance reference. All hyperparameter
configurations and implementation details are provided in Appendix, and all
experiments are fully reproducible using the code released in the accompanying
repository.

\subsection{Synthetic Benchmarks}

We construct a controlled synthetic benchmark suite using a $3 \times 3$
factorial design that crosses three trend patterns with three seasonality
regimes, yielding nine datasets with known ground-truth components.
Each time series contains 2000 observations with additive Gaussian noise
($\sigma = 1.0$).

\paragraph{Trend patterns.}
We consider three representative trend structures:
\textbf{Linear} trends capturing monotonic growth,
\textbf{Inverted-V} trends modeling rise-and-fall dynamics, and
\textbf{Piecewise} trends with abrupt regime shifts that represent structural
breaks.

\paragraph{Seasonality regimes.}
Each trend pattern is combined with one of three seasonal behaviors:
\textbf{Fixed} seasonality with stationary cycles,
\textbf{Transitive} seasonality characterized by abrupt mid-series changes, and
\textbf{Variable} seasonality with continuously drifting cycle lengths that
emulate quasi-periodic real-world phenomena.

\paragraph{Benchmark design rationale.}
This factorial construction isolates complementary sources of temporal
complexity—trend nonlinearity and seasonal nonstationarity—enabling systematic
analysis of robustness, adaptability, and failure modes across competing
decomposition methods.

\subsection{Real-World Benchmarks}

\textbf{ETDataset} \cite{haoyietal-informer-2021} (ETTh1, ETTh2) consists of
hourly electricity transformer load measurements characterized by overlapping
daily and weekly seasonalities, long-term consumption trends, and regime shifts
arising from evolving usage patterns. These multi-scale periodic structures
challenge methods relying on a single fixed periodicity.

\textbf{Sunspot} \cite{silso2015sunspot} records historical solar activity with
quasi-periodic cycles of irregular length (8--15 years) and varying amplitude.
The absence of stable periodicity provides a stringent test for decomposition
methods under fundamentally non-stationary seasonality.

\subsection{Evaluation Protocol}

\paragraph{Synthetic datasets.}
For synthetic datasets with known ground-truth components, decomposition accuracy
is evaluated quantitatively using mean absolute error (MAE)\footnote{
The MAE is defined as
$
\text{MAE} = \frac{1}{n} \sum_{i=1}^{n} \left| \hat{y}_i - y_i \right|,
$
where $y_i$ and $\hat{y}_i$ denote the ground-truth and estimated component values,
respectively.
}, computed separately for the trend, seasonal, and residual components.
We additionally report an overall MAE obtained by averaging errors across all
components.

\paragraph{Real-world datasets.}
For real-world datasets, where ground-truth decompositions are unavailable,
evaluation relies on qualitative inspection and quantitative residual diagnostics.
Specifically, we assess remaining temporal dependence in the residual component
using the Ljung--Box portmanteau test\footnote{
The Ljung--Box test evaluates the null hypothesis that a
time series exhibits no autocorrelation up to a specified lag $h$.
The test statistic is
$\textstyle Q_h = n(n+2)\sum_{k=1}^{h}\allowbreak \hat{\rho}_k^2/(n-k)$,
where $n$ is the sample size and $\hat{\rho}_k$ is the sample autocorrelation at
lag $k$.
Under the null hypothesis of serial independence, $Q_h$ asymptotically follows a
$\chi^2$ distribution with $h$ degrees of freedom.
}~\cite{ljung1978measure}.
We report Ljung--Box statistics at multiple lags ($Q_{10}$, $Q_{20}$, and $Q_{30}$),
where lower values indicate weaker residual serial dependence and more effective
removal of structured temporal dynamics.

\section{Results}
\label{sec:results}

\subsection{Synthetic Data results}

We evaluate decomposition accuracy on synthetic datasets that vary trend shape,
season-length behavior, and noise structure. Performance is measured using mean
absolute error (MAE) on the trend, seasonal, and residual components, as well as
overall reconstruction error.

\begin{table}[htbp]
\centering
\setlength{\tabcolsep}{6pt}
\renewcommand{\arraystretch}{1}
\caption{Average MAE Across All Datasets}
\vspace{-4pt}
\label{tb:SynOverallPer}
\begin{tabular}{lcccc}
\toprule
\textbf{Model} & \textbf{Trend} & \textbf{Seasonal} & \textbf{Residual} & \textbf{Overall} \\
\midrule
STL & 4.07 & 16.46 & 13.64 & 11.39 \\
STR & 20.64 & 23.90 & 4.21 & 16.25 \\
FastRobustSTL & 20.80 & 21.20 & \textbf{1.12} & 14.37 \\
$ASTD$ & 21.02 & 16.27 & 11.28 & 16.19 \\
OnlineSTL & 8.51 & 6.89 & 5.40 & 6.93 \\
OneShotSTL & 33.52 & 32.20 & 7.75 & 24.49 \\
$ASTD_{Online}$ & 7.57 & 10.90 & 11.41 & 9.96 \\
LGTD & \textbf{3.62} & \textbf{4.84} & 3.55 & \textbf{4.00} \\
\bottomrule
\end{tabular}
\end{table}

\paragraph{Overall performance.}
Table~\ref{tb:SynOverallPer} reports MAE averaged across all synthetic datasets.
\textsc{LGTD} achieves the lowest overall MAE, with the best trend and seasonal
accuracy.
While some baselines perform well on individual components, they incur larger
errors elsewhere, leading to higher overall error.
In contrast, \textsc{LGTD} maintains balanced accuracy across all components
without requiring a predefined season length.

Notably, LGTD does not optimize specifically for any single component, but instead minimizes cross-component interference, leading to lower overall error.

\begin{table}[htbp]
\centering
\setlength{\tabcolsep}{6pt}
\renewcommand{\arraystretch}{1}
\caption{Average MAE for Transitive Period Datasets}
\vspace{-4pt}
\label{tb:SynTrans}
\begin{tabular}{lcccc}
\toprule
\textbf{Model} & \textbf{Trend} & \textbf{Seasonal} & \textbf{Residual} & \textbf{Overall} \\
\midrule
STL & \textbf{1.44} & 12.32 & 12.31 & 8.69 \\
STR & 25.26 & 32.69 & 9.24 & 22.40 \\
FastRobustSTL & 25.11 & 25.56 & \textbf{1.36} & 17.35 \\
$ASTD$ & 21.00 & 16.18 & 10.51 & 15.89 \\
OnlineSTL & 9.16 & 7.37 & 5.32 & 7.28 \\
OneShotSTL & 55.85 & 54.70 & 5.11 & 38.56 \\
$ASTD_{Online}$ & 7.51 & 9.71 & 10.17 & 9.13 \\
LGTD & 2.04 & \textbf{4.06} & 2.83 & \textbf{2.98} \\
\bottomrule
\end{tabular}
\end{table}

\paragraph{Transitive season-length datasets.}
Results for datasets with stable season length are shown in
Table~\ref{tb:SynTrans}.
STL performs well on trend estimation but exhibits higher seasonal and residual
errors.
Adaptive and online methods improve robustness, yet still depend on explicit or
estimated season length.
\textsc{LGTD} attains the lowest overall MAE with consistently low component-wise
error, demonstrating competitive performance even when periodic structure is
stable.

\begin{figure}[t]
    \centering
    \includegraphics[width=\linewidth]{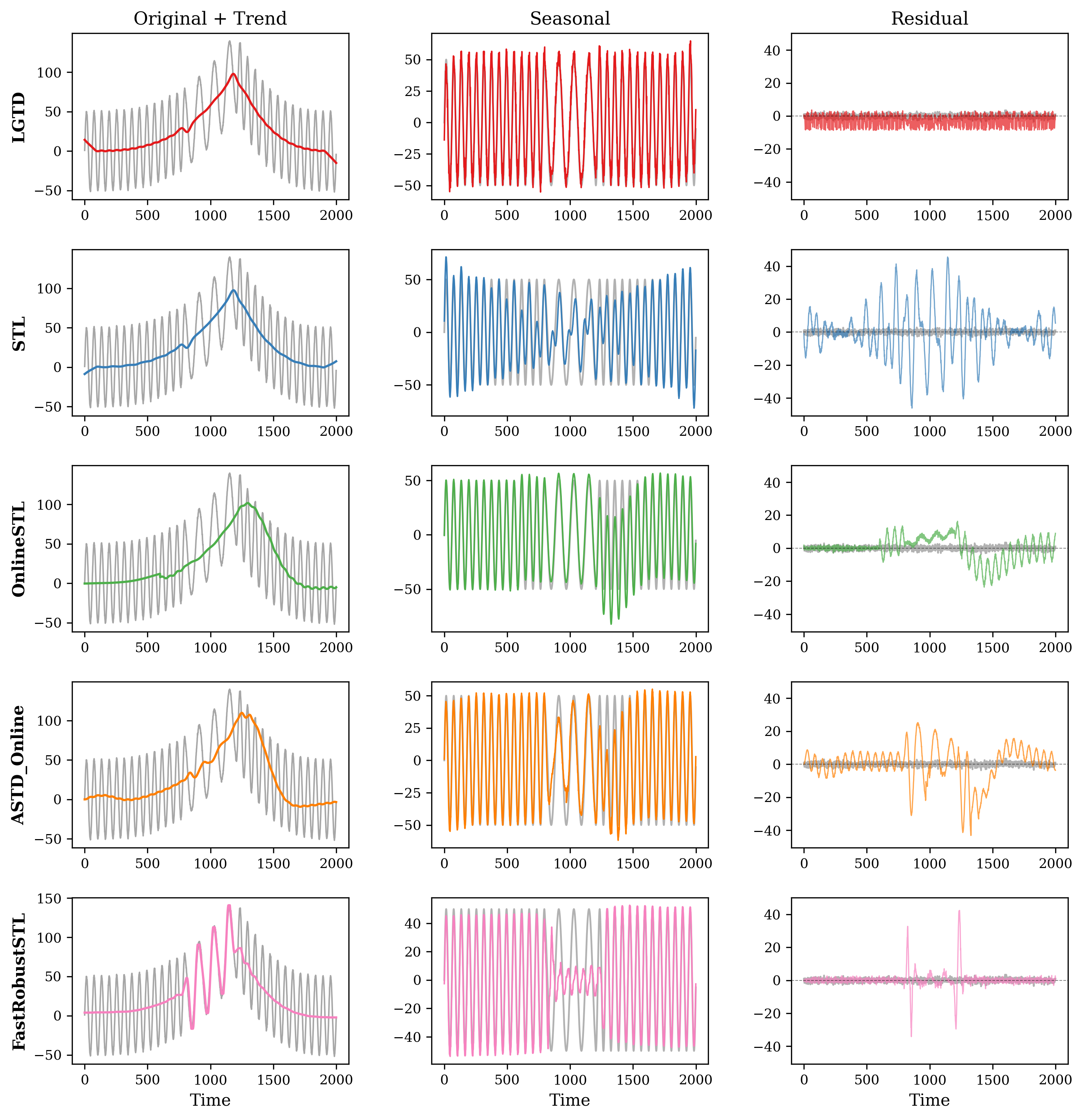}
    \caption{Decomposition on synthetic transitive season-length data, where LGTD preserves smooth trends and coherent emergent seasonality under changing cycles, while period-based methods exhibit distortion and residual leakage.}
    \label{fig:SyntrnsRes}
\end{figure}

Fig.~\ref{fig:SyntrnsRes} illustrates decomposition on transitive season-length
data, where the underlying seasonal structure gradually shifts.
LGTD adapts smoothly to these transitions, preserving a stable global trend and
low-structure residuals, while period-based methods exhibit seasonal distortion
and increased residual variance.

\begin{table}[htbp]
\centering
\setlength{\tabcolsep}{6pt}
\renewcommand{\arraystretch}{1}
\caption{Average MAE for Variable Period Datasets}
\vspace{-4pt}
\label{tb:SynVar}
\begin{tabular}{lcccc}
\toprule
\textbf{Model} \footnotesize& \textbf{Trend} & \textbf{Seasonal} & \textbf{Residual} & \textbf{Overall} \\
\midrule
STL & 10.57 & 36.77 & 28.19 & 25.17 \\
STR & 35.22 & 37.40 & 2.60 & 25.08 \\
FastRobustSTL & 37.00 & 37.13 & \textbf{1.08} & 25.07 \\
$ASTD$ & 21.02 & 17.80 & 14.15 & 17.66 \\
OneShotSTL & 40.65 & 38.22 & 17.10 & 31.99 \\
OnlineSTL & 10.49 & 8.79 & 7.38 & 8.89 \\
$ASTD_{Online}$ & 8.40 & 15.11 & 15.76 & 13.09 \\
LGTD & \textbf{7.63} & \textbf{7.78} & 5.33 & \textbf{6.91} \\
\bottomrule
\end{tabular}
\end{table}

\paragraph{Variable season-length datasets.}
Table~\ref{tb:SynVar} summarizes results for datasets with drifting or irregular
season length.
Here, the performance gap widens: methods relying on fixed or estimated season
length show marked degradation in seasonal and overall MAE.
OnlineSTL remains relatively robust, but still exhibits higher trend and
seasonal errors than \textsc{LGTD}.
\textsc{LGTD} achieves the lowest overall MAE and best seasonal accuracy,
highlighting the robustness of its period-free formulation under nonstationary
temporal patterns.

Fig.~\ref{fig:SynVarRes} further shows that LGTD consistently recovers smooth
global trends while adapting seasonal structure to changing cycle length and
amplitude.
In contrast, period-based methods exhibit component leakage and higher-variance
residuals as season length drifts.

Overall, the synthetic experiments demonstrate that \textsc{LGTD} provides stable
and balanced decomposition accuracy across both stable and drifting season-length
regimes, supporting its use as a low-touch decomposition primitive in
heterogeneous settings.

\begin{figure}[t]
    \centering
    \includegraphics[width=\linewidth]{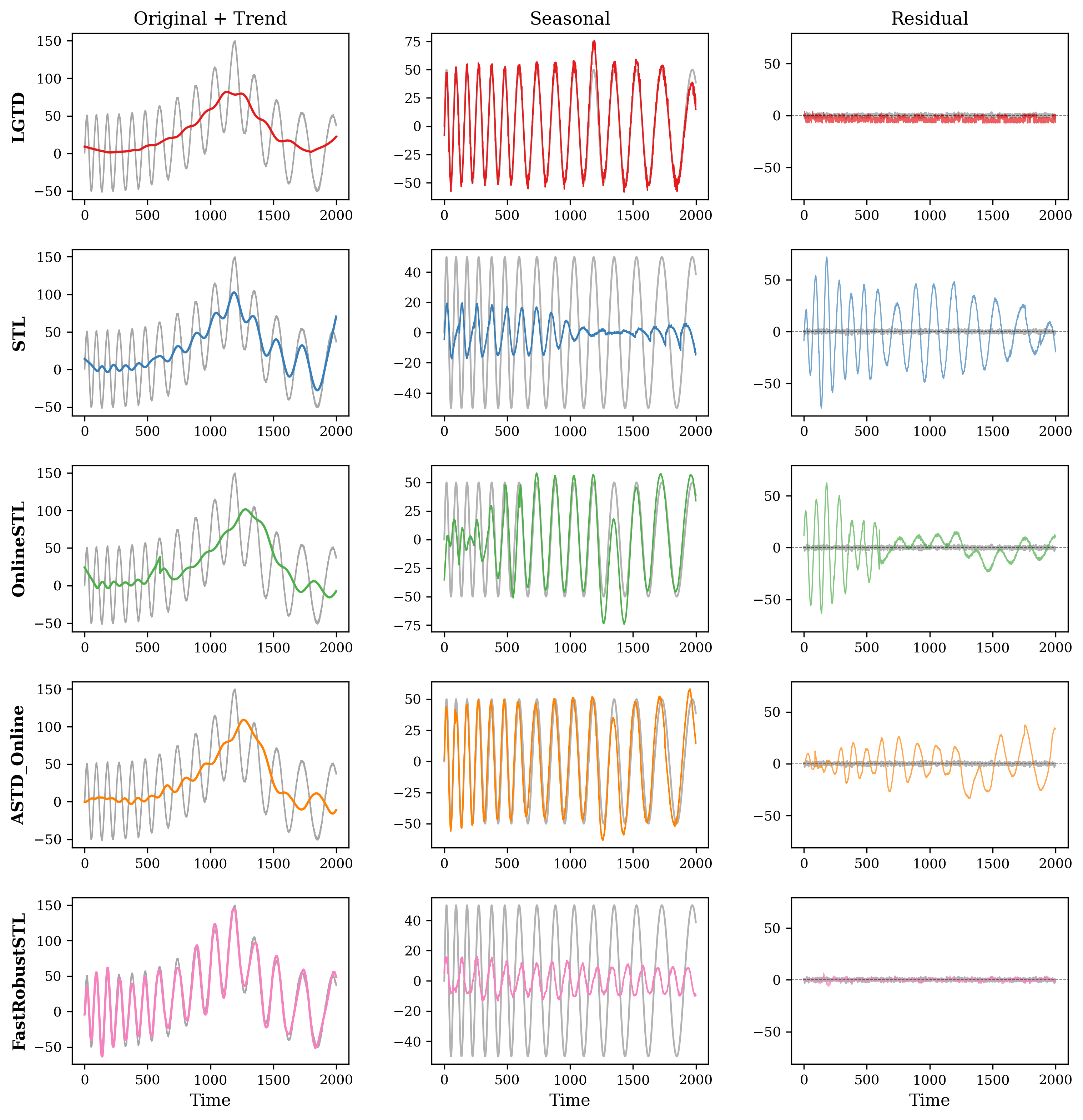}
    \caption{Decomposition results on synthetic variable season-length data, showing that LGTD robustly captures the evolving global trend and adapts to changing seasonal structure via emergent local trends, while maintaining low-variance residuals compared to period-based baselines that struggle under nonstationary season lengths.}
    \label{fig:SynVarRes}
\end{figure}

In the variable season-length setting, Fig.~\ref{fig:SynVarRes} shows LGTD consistently recovers a smooth global trend while adapting the seasonal component to gradual changes in cycle length and amplitude, without relying on a fixed seasonal assumption. In contrast, period-based methods exhibit leakage between seasonal and residual components when the season length drifts, leading to less stable trend estimates and higher-variance residuals.

\subsection{Real-World Data Results}
We evaluate \textsc{LGTD} on real-world datasets with complex trends, drifting
cycles, and nonstationary dynamics. Performance is assessed through qualitative
inspection of component separation and quantitative residual diagnostics using
the Ljung--Box portmanteau test.

\begin{figure}[t]
    \centering
    \includegraphics[width=\linewidth]{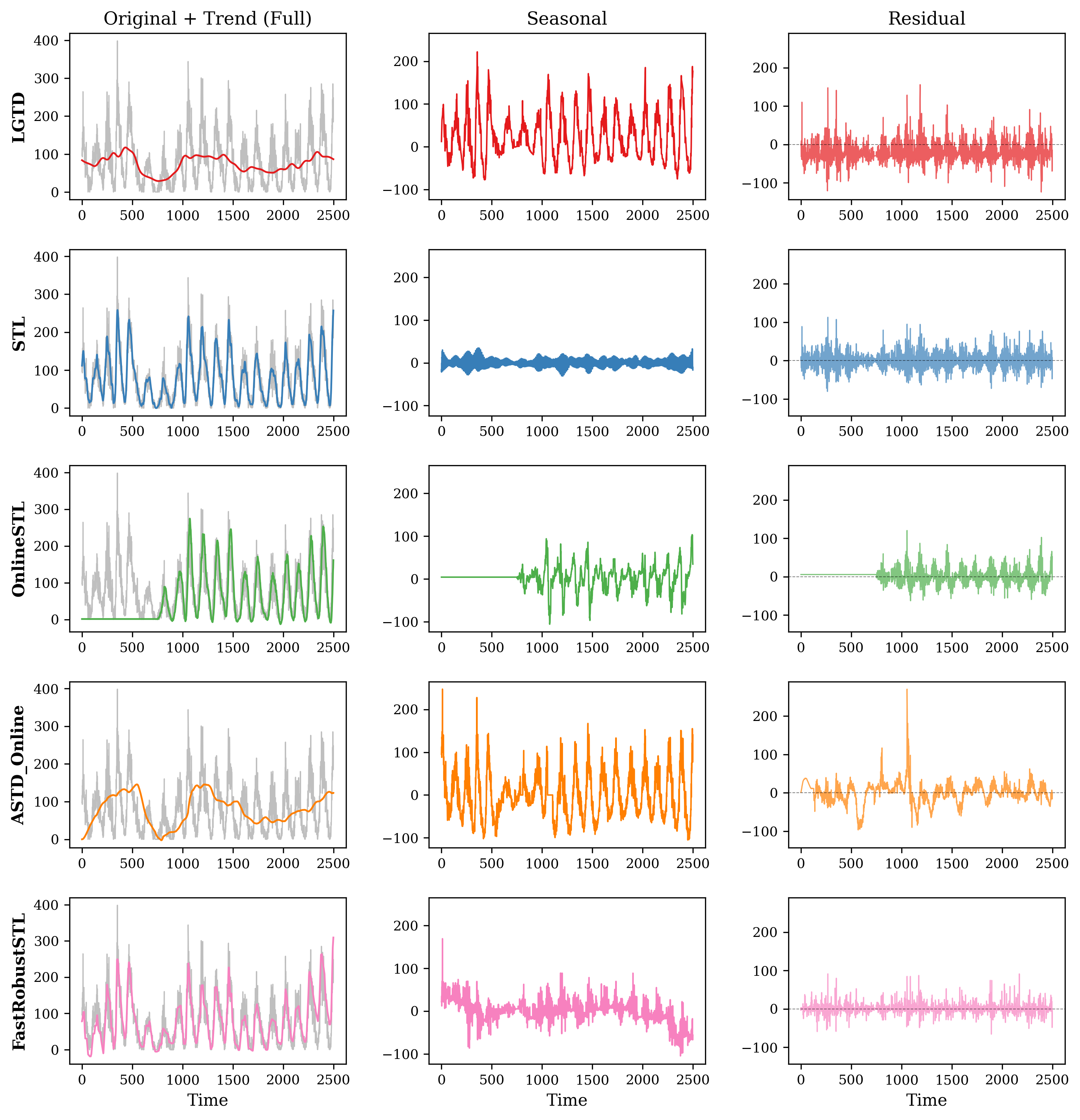}
    \caption{Decomposition of the SILSO sunspot series, showing that LGTD captures long-term solar activity trends and quasi-periodic emergent seasonality with drifting cycles, while producing low-structure residuals without assuming a fixed season length.}
    \label{fig:SunspotRes}
\end{figure}

Fig.~\ref{fig:SunspotRes} shows decompositions on the SILSO sunspot series, which
exhibits quasi-periodic behavior with drifting cycle length and amplitude.
\textsc{LGTD} recovers meaningful long-term trends and emergent seasonal
structure without imposing a fixed periodic template, while more effectively
suppressing residual autocorrelation than period-based baselines.

\begin{figure}[t]
    \centering
    \includegraphics[width=\linewidth]{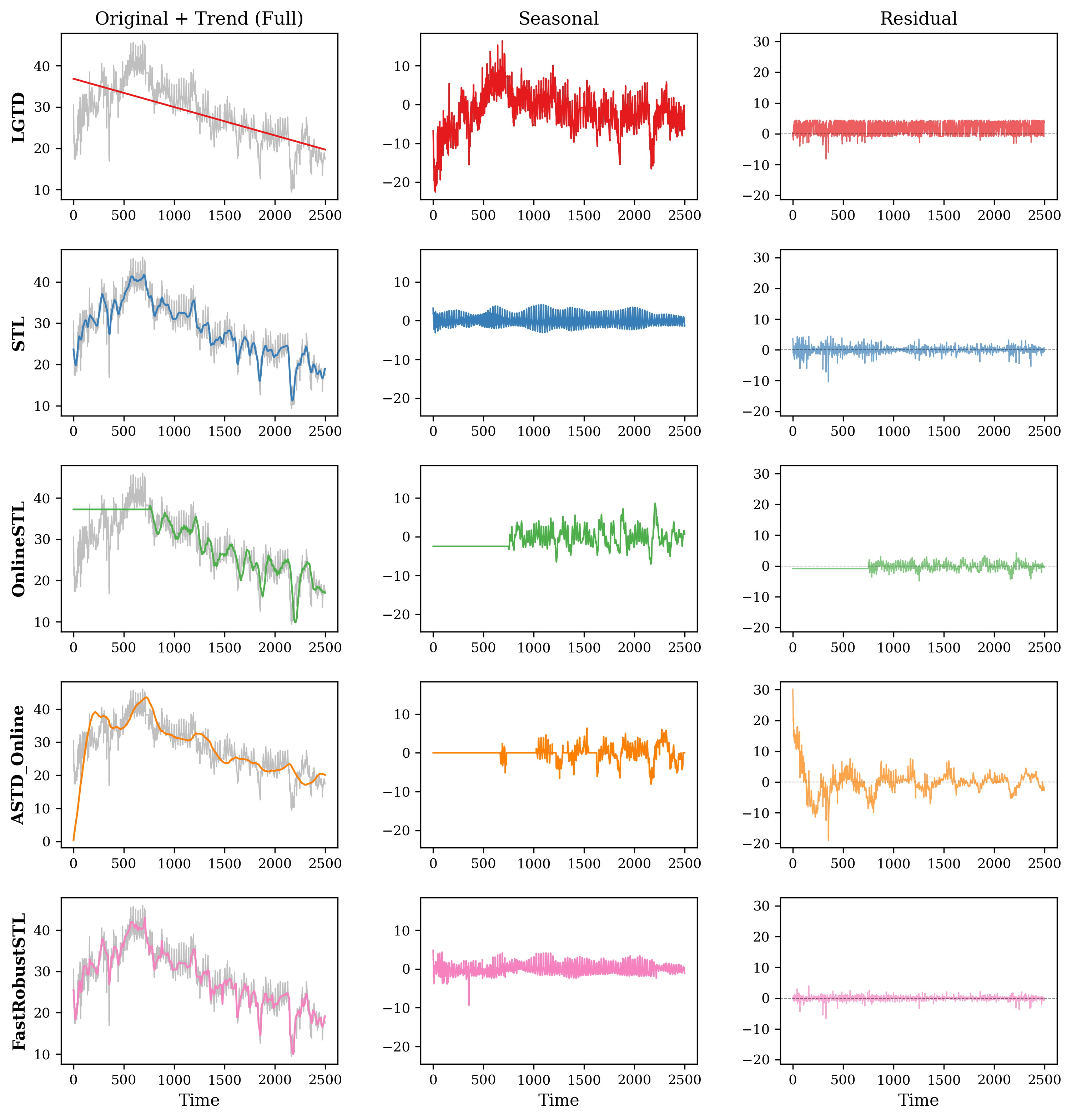}
    \caption{Decomposition of the ETTh1 dataset, illustrating that LGTD extracts a smooth global trend and structured emergent seasonality while leaving low-variance residuals, without requiring a predefined season length.}
    \label{fig:ETTh1Res}
\end{figure}

Fig.~\ref{fig:ETTh1Res} presents results on the ETTh1 electricity transformer
dataset (ETTh2 exhibits similar behavior).
\textsc{LGTD} extracts smooth global trends and coherent emergent seasonality
capturing dominant consumption patterns, while adapting to local variations in
seasonal strength and regime shifts.
This yields residuals with low variance and minimal remaining temporal structure
compared to STL-based methods.

\begin{figure}[t]
    \centering
    \includegraphics[width=\columnwidth]{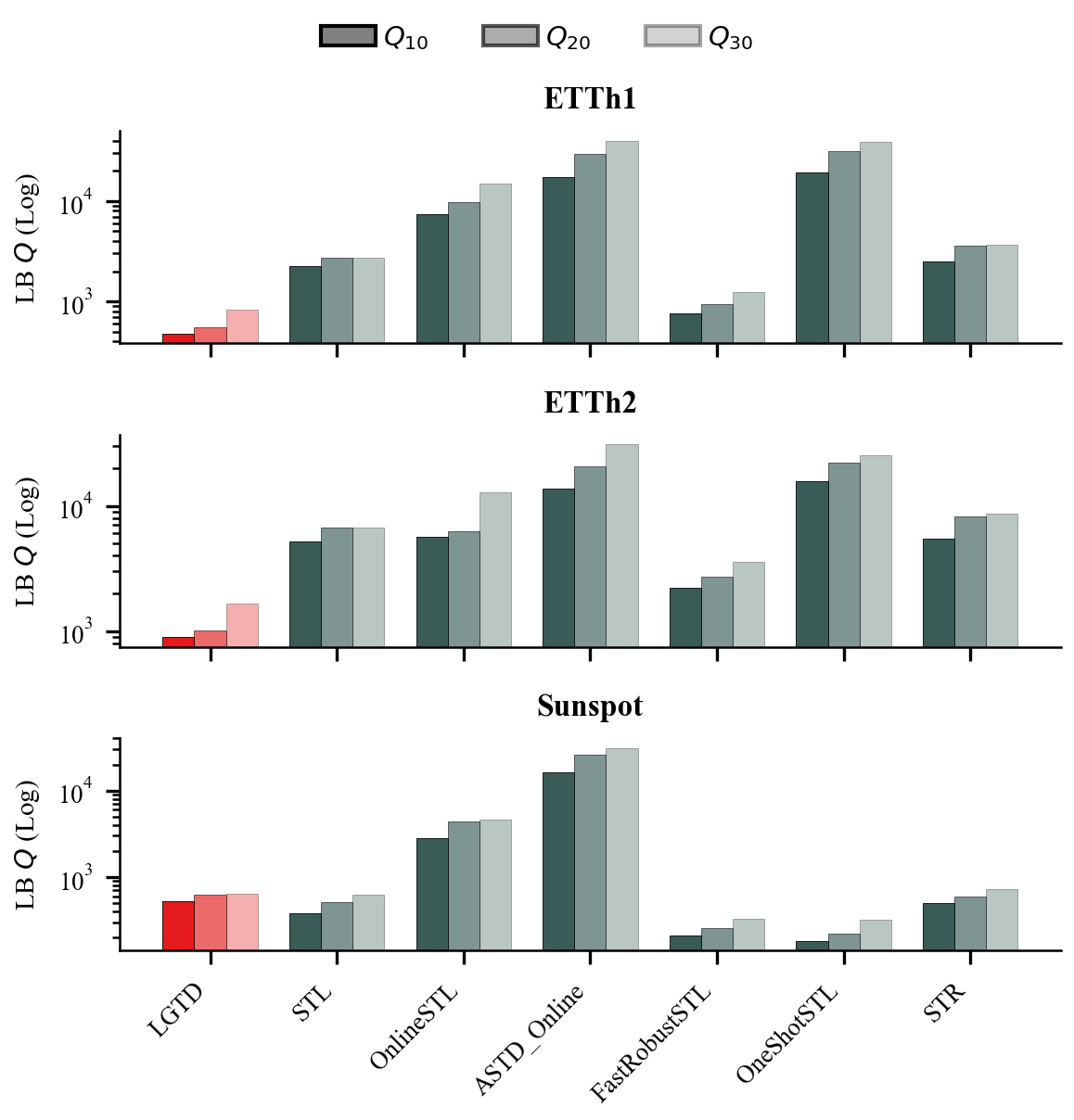}
    \caption{Ljung--Box portmanteau test statistics computed on residual components obtained from different trend--seasonal decomposition methods across the ETTh1, ETTh2, and Sunspot datasets. The statistics are evaluated at lags $10$, $20$, and $30$ ($Q_{10}$, $Q_{20}$, and $Q_{30}$) and reported on a logarithmic scale. Lower $Q$ values indicate weaker residual serial dependence, suggesting more effective removal of temporal structure by the decomposition method.}
    \label{fig:RealQuant}
\end{figure}

Fig.~\ref{fig:RealQuant} reports Ljung--Box statistics on residuals across ETTh1,
ETTh2, and Sunspot.
Across all datasets and evaluated lags, \textsc{LGTD} consistently attains the
lowest $Q$ values, indicating minimal residual serial dependence after
decomposition.
In contrast, period-based and online methods exhibit substantially higher $Q$
statistics, particularly at larger lags, revealing persistent temporal structure
in the residuals.

Overall, these results demonstrate that \textsc{LGTD} more effectively isolates
stochastic components in real-world time series with irregular patterns and
drifting dynamics.

\subsection{Scalability}
\label{sec:scalability}

Beyond decomposition quality, a practical decomposition primitive must scale to
long series with predictable cost. We therefore measure wall-clock runtime and
peak memory as a function of series length $N$. Synthetic series (linear trend,
fixed-period sinusoid, Gaussian noise) are generated for
$N \in \{10^3,\dots,10^5\}$, and each method's \texttt{fit\_transform} is timed
in isolation. Period-based methods receive a fixed season length ($P=120$)
independent of $N$---their best case---while \textsc{LGTD} and ASTD use no period.
For each $(\text{method}, N)$ we discard one warmup run and report the median of
five timed runs to mitigate scheduling variation; memory is the peak Python-side
allocation measured separately via \texttt{tracemalloc}. To fix the configuration
under study, \textsc{LGTD} uses a linear global trend, since its
\textsc{AutoTrend-LLT} stage---which carries the $O(K\,T)$ bound of
Section~\ref{subsec:complexity}---is independent of the global-trend model. All
experiments were run on a 2024 MacBook Air (13-inch) with an Apple M3 chip
(8-core CPU) and 16\,GB of unified memory, under macOS. RobustSTL and FastRobustSTL are omitted from the scaling study; both are
robust optimization-based methods with super-linear cost (per-iteration
$O(N^2)$ for RobustSTL~\cite{10.1609/aaai.v33i01.33015409}, reduced to
$O(N\log N)$ by FastRobustSTL~\cite{10.1145/3394486.3403271}), placing them
in the same regime already demonstrated by STR; they therefore add no further
information about asymptotic scaling.

\begin{table}[htbp]
\centering
\setlength{\tabcolsep}{5pt}
\renewcommand{\arraystretch}{1}
\caption{Runtime (s, median) at two series lengths and peak memory at
$N{=}10^{5}$. ``---'' denotes runs that did not complete within the time budget.}
\vspace{-4pt}
\label{tb:Scalability}
\begin{tabular}{lccc}
\toprule
\textbf{Model} & \textbf{Time @ $10^{4}$} & \textbf{Time @ $10^{5}$} & \textbf{Mem @ $10^{5}$ (MB)} \\
\midrule
STL                 & 0.380   & 3.761  & 13.0 \\
STR                 & 214.85  & ---    & ---  \\
$ASTD$              & 1.216   & 76.455 & 16.9 \\
OnlineSTL           & 1.056   & 10.112 & 8.9  \\
OneShotSTL$^\dagger$ & 0.328   & 0.790  & ---  \\
LGTD                & \textbf{0.009} & \textbf{0.080} & 17.5 \\
\bottomrule
\end{tabular}
\end{table}

\begin{figure*}[t]
    \centering
    \includegraphics[width=\textwidth]{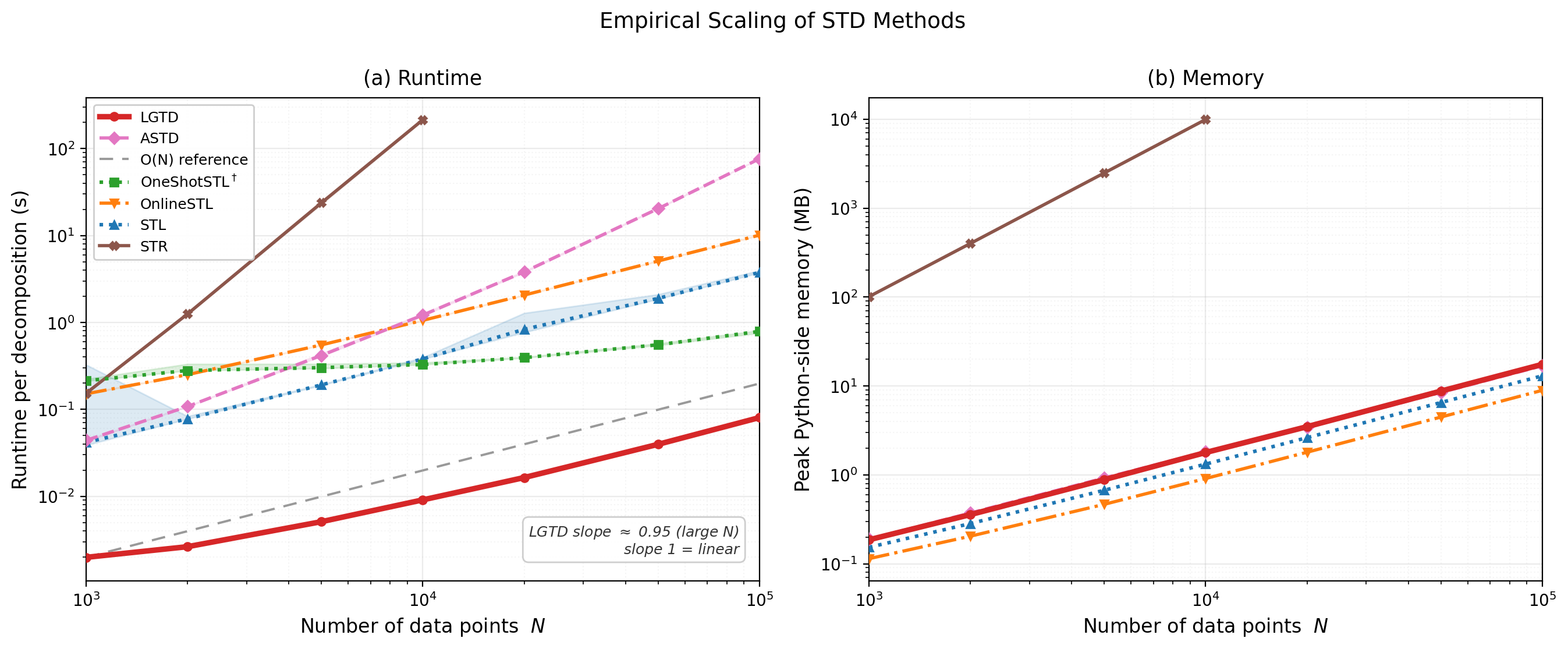}
    \caption{Empirical scaling with series length $N$ (log--log axes; slope $1$
    corresponds to linear growth). \textbf{(a)} Runtime per decomposition.
    \textsc{LGTD} tracks the $O(N)$ reference (fitted large-$N$ slope $\approx
    0.95$) and is the fastest method across the entire range, while STR and ASTD
    grow super-linearly. \textbf{(b)} Peak Python-side memory; \textsc{LGTD}, STL,
    and OnlineSTL grow linearly, whereas STR's global regression reaches roughly
    $10$\,GB already at $N{=}10^{4}$. $^\dagger$OneShotSTL runs as a Java
    subprocess; its runtime is dominated by per-call JVM startup at these lengths
    and is therefore not directly comparable to the in-process methods, and it is
    omitted from~(b) because \texttt{tracemalloc} cannot observe the JVM heap.}
    \label{fig:Scalability}
\end{figure*}

\paragraph{Runtime.}
Fig.~\ref{fig:Scalability}(a) and Table~\ref{tb:Scalability} show that
\textsc{LGTD} scales linearly in $N$, with a fitted log--log slope of $0.95$ over
the large-$N$ regime, empirically confirming the $O(K\,T)$ analysis of
Section~\ref{subsec:complexity}. \textsc{LGTD} is the fastest method throughout: at
$N{=}10^{5}$ it is roughly $47\times$ faster than STL, $126\times$ faster than
OnlineSTL, and nearly three orders of magnitude faster than ASTD. Although STL is
also linear in $N$, \textsc{LGTD} attains a substantially smaller constant and,
unlike STL, requires no season-length parameter. The clearest separation is
asymptotic rather than constant-factor: STR and ASTD scale super-linearly, with
STR failing to complete $N{=}10^{5}$ within the time budget---already at
$N{=}10^{4}$ it requires over $200$\,s---reflecting the cost of its global
regression formulation.

\paragraph{Memory.}
Fig.~\ref{fig:Scalability}(b) shows that \textsc{LGTD} uses memory linear in $N$
($\approx 17$\,MB at $N{=}10^{5}$), comparable to STL and OnlineSTL, as expected
for a batch method that stores the input together with its component arrays. STR
is again the outlier, with peak memory near $10$\,GB at $N{=}10^{4}$, mirroring its
super-linear runtime. These results indicate that \textsc{LGTD} combines balanced
decomposition accuracy with linear time and memory and no dependence on a
season-length parameter, supporting its use as a low-touch primitive on long,
heterogeneous series.

\subsection{Parameter Sensitivity}
\textsc{LGTD} includes a small number of hyperparameters associated with local
trend inference, specifically the window size \( W \) and the error percentile
\( p \). Sensitivity analysis indicates that \textsc{LGTD} maintains stable
decomposition performance across a wide range of parameter settings, with error
varying smoothly and without abrupt degradation. The default configuration lies
within a broad low-error region across diverse data regimes, suggesting that
effective performance does not rely on precise hyperparameter tuning. This
behavior is consistent with the design of AutoTrend-LLT, where trend assignments
are driven by relative error ranking and iteratively refined, yielding
well-conditioned behavior with respect to parameter choice. Detailed sensitivity
results are provided in Appendix.%Appendix~\ref{app:sensitivity}.

\subsection{Discussion}
The results demonstrate that \textsc{LGTD} addresses a central limitation of
period-based seasonal--trend decomposition methods: sensitivity to explicit or
estimated season length. Across synthetic datasets, \textsc{LGTD} achieves the
lowest or near-lowest overall MAE, with particularly strong gains in transitive
and variable season-length regimes where periodic assumptions break down
(Tables~\ref{tb:SynTrans} and~\ref{tb:SynVar}). In contrast, methods relying on
fixed or adaptive periods exhibit increased seasonal distortion and higher
overall error under nonstationary conditions.

By modeling seasonality as an emergent property of recurring local trend regimes,
\textsc{LGTD} maintains balanced accuracy across trend, seasonal, and residual
components, reducing cross-component leakage and yielding more stable
decompositions, as reflected in both quantitative metrics and visual results
(Figs.~\ref{fig:SyntrnsRes} and~\ref{fig:SynVarRes}). Real-world experiments on the
Sunspot and ETTh datasets further confirm this behavior: \textsc{LGTD} extracts
coherent trends and adaptive seasonal structure while producing residuals with
minimal serial dependence, as verified by the Ljung--Box diagnostics in
Fig.~\ref{fig:RealQuant}.

These accuracy benefits are achieved without sacrificing efficiency. As shown in
Section~\ref{sec:scalability}, \textsc{LGTD} scales linearly in series length in
both runtime and memory (Fig.~\ref{fig:Scalability}), remaining the fastest method
across all tested lengths while several baselines degrade super-linearly or exhaust
memory. Together with its lack of any season-length parameter, this makes
\textsc{LGTD} well suited as a low-touch decomposition primitive for long,
heterogeneous time series where per-series period tuning is impractical.

\paragraph{Limitations.}
\textsc{LGTD} prioritizes robustness to nonstationary and irregular temporal
structure by operating without explicit seasonal priors. In settings where
seasonality is strictly periodic and known in advance, classical methods such as
STL may achieve comparable trend accuracy with simpler parameterization.
Additionally, representing seasonality via regime recurrence rather than explicit
frequency components favors adaptive decomposition over direct harmonic
interpretability. Finally, \textsc{LGTD} is formulated as a batch method with
linear time and memory in the series length; while this is efficient, streaming
decomposition methods that maintain bounded per-update state may be preferable in
strictly online settings, and extending the emergent-seasonality formulation to an
incremental variant is a natural direction for future work. These trade-offs
reflect a deliberate design choice toward generality and stability across
heterogeneous time series.

\section{Conclusion}
We introduced \textsc{LGTD}, a season-length–free time-series decomposition
framework that models seasonality as an emergent consequence of recurring local
trend regimes rather than an explicit periodic prior. By combining smooth global
trend estimation with adaptive local linear trend inference, \textsc{LGTD}
achieves stable and balanced decomposition across fixed, transitive, and variable
season-length settings. Extensive experiments on synthetic and real-world
datasets demonstrate that \textsc{LGTD} consistently outperforms period-based
baselines under nonstationary and irregular temporal dynamics. At the same time,
\textsc{LGTD} scales linearly in both runtime and memory with series length and
remains the fastest method across all tested lengths, whereas several baselines
degrade super-linearly or exhaust memory. Together with its lack of any
season-length parameter, these properties position \textsc{LGTD} as a robust,
efficient, and interpretable decomposition primitive for long, heterogeneous
time-series analysis. A natural direction for future work is to extend the
emergent-seasonality formulation to an incremental, streaming variant for strictly
online settings. An implementation and datasets, as well as Supplementary Materials are provided in Appendix that can be found in \url{https://github.com/chotanansub/LGTD}.
%Appendix~\ref{app:code}.

\clearpage

\balance
\bibliographystyle{ACM-Reference-Format}
\bibliography{KDD2026}

@inproceedings{10.1609/aaai.v33i01.33015409,
author = {Wen, Qingsong and Gao, Jingkun and Song, Xiaomin and Sun, Liang and Xu, Huan and Zhu, Shenghuo},
title = {RobustSTL: a robust seasonal-trend decomposition algorithm for long time series},
year = {2019},
isbn = {978-1-57735-809-1},
publisher = {AAAI Press},
url = {https://doi.org/10.1609/aaai.v33i01.33015409},
doi = {10.1609/aaai.v33i01.33015409},
booktitle = {Proceedings of the Thirty-Third AAAI Conference on Artificial Intelligence and Thirty-First Innovative Applications of Artificial Intelligence Conference and Ninth AAAI Symposium on Educational Advances in Artificial Intelligence},
articleno = {663},
numpages = {8},
location = {Honolulu, Hawaii, USA},
series = {AAAI'19/IAAI'19/EAAI'19}
}

@inproceedings{10.1145/3394486.3403271,
author = {Wen, Qingsong and Zhang, Zhe and Li, Yan and Sun, Liang},
title = {Fast RobustSTL: Efficient and Robust Seasonal-Trend Decomposition for Time Series with Complex Patterns},
year = {2020},
isbn = {9781450379984},
publisher = {Association for Computing Machinery},
address = {New York, NY, USA},
url = {https://doi.org/10.1145/3394486.3403271},
doi = {10.1145/3394486.3403271},
booktitle = {Proceedings of the 26th ACM SIGKDD International Conference on Knowledge Discovery \& Data Mining},
pages = {2203–2213},
numpages = {11},
location = {Virtual Event, CA, USA},
series = {KDD '20}
}

@article{cleveland1990stl,
  title={STL: A seasonal-trend decomposition},
  author={Cleveland, Robert B and Cleveland, William S and McRae, Jean E and Terpenning, Irma and others},
  journal={J. off. Stat},
  volume={6},
  number={1},
  pages={3--73},
  year={1990}
}

@article{dokumentov2022str,
  title={STR: Seasonal-trend decomposition using regression},
  author={Dokumentov, Alexander and Hyndman, Rob J},
  journal={INFORMS Journal on Data Science},
  volume={1},
  number={1},
  pages={50--62},
  year={2022},
  publisher={INFORMS}
}

@article{10.14778/3523210.3523219,
author = {Mishra, Abhinav and Sriharsha, Ram and Zhong, Sichen},
title = {OnlineSTL: scaling time series decomposition by 100x},
year = {2022},
issue_date = {March 2022},
publisher = {VLDB Endowment},
volume = {15},
number = {7},
issn = {2150-8097},
url = {https://doi.org/10.14778/3523210.3523219},
doi = {10.14778/3523210.3523219},
journal = {Proc. VLDB Endow.},
month = mar,
pages = {1417–1425},
numpages = {9}
}

@article{10.14778/3583140.3583155,
author = {He, Xiao and Li, Ye and Tan, Jian and Wu, Bin and Li, Feifei},
title = {OneShotSTL: One-Shot Seasonal-Trend Decomposition For Online Time Series Anomaly Detection And Forecasting},
year = {2023},
issue_date = {February 2023},
publisher = {VLDB Endowment},
volume = {16},
number = {6},
issn = {2150-8097},
url = {https://doi.org/10.14778/3583140.3583155},
doi = {10.14778/3583140.3583155},
journal = {Proc. VLDB Endow.},
month = feb,
pages = {1399–1412},
numpages = {14}
}

@inproceedings{10.1007/978-3-031-70344-7_25,
author = {Phungtua-eng, Thanapol and Yamamoto, Yoshitaka},
title = {Adaptive Seasonal-Trend Decomposition for Streaming Time Series Data with Transitions and Fluctuations in Seasonality},
year = {2024},
isbn = {978-3-031-70343-0},
publisher = {Springer-Verlag},
address = {Berlin, Heidelberg},
url = {https://doi.org/10.1007/978-3-031-70344-7_25},
doi = {10.1007/978-3-031-70344-7_25},
booktitle = {Machine Learning and Knowledge Discovery in Databases. Research Track: European Conference, ECML PKDD 2024, Vilnius, Lithuania, September 9–13, 2024, Proceedings, Part II},
pages = {426–443},
numpages = {18},
location = {Vilnius, Lithuania}
}

@inproceedings{10.1145/3637528.3671510,
author = {Wang, Haoyu and Guo, Hongke and Zhu, Zhaoliang and Zhang, You and Zhou, Yu and Zheng, Xudong},
title = {BacktrackSTL: Ultra-Fast Online Seasonal-Trend Decomposition with Backtrack Technique},
year = {2024},
isbn = {9798400704901},
publisher = {Association for Computing Machinery},
address = {New York, NY, USA},
url = {https://doi.org/10.1145/3637528.3671510},
doi = {10.1145/3637528.3671510},
booktitle = {Proceedings of the 30th ACM SIGKDD Conference on Knowledge Discovery and Data Mining},
pages = {5848–5859},
numpages = {12},
location = {Barcelona, Spain},
series = {KDD '24}
}

@INPROCEEDINGS{11112870,
  author={Chen, Zijie and Song, Shaoxu and Wang, Jianmin},
  booktitle={2025 IEEE 41st International Conference on Data Engineering (ICDE)}, 
  title={OneRoundSTL: In-Database Seasonal-Trend Decomposition}, 
  year={2025},
  volume={},
  number={},
  pages={1-13},
  doi={10.1109/ICDE65448.2025.00060},
  ISSN={2375-026X},
  month={May},}

@misc{bandara2021mstlseasonaltrenddecompositionalgorithm,
      title={MSTL: A Seasonal-Trend Decomposition Algorithm for Time Series with Multiple Seasonal Patterns}, 
      author={Kasun Bandara and Rob J Hyndman and Christoph Bergmeir},
      year={2021},
      eprint={2107.13462},
      archivePrefix={arXiv},
      primaryClass={stat.AP},
      url={https://arxiv.org/abs/2107.13462}, 
}

@article{10.1137/070690274,
author = {Kim, Seung-Jean and Koh, Kwangmoo and Boyd, Stephen and Gorinevsky, Dimitry},
title = {$\ell_1$ Trend Filtering},
year = {2009},
issue_date = {May 2009},
publisher = {Society for Industrial and Applied Mathematics},
address = {USA},
volume = {51},
number = {2},
issn = {0036-1445},
url = {https://doi.org/10.1137/070690274},
doi = {10.1137/070690274},
journal = {SIAM Rev.},
month = may,
pages = {339–360},
numpages = {22}
}

@inproceedings{haoyietal-informer-2021,
  author    = {Haoyi Zhou and
               Shanghang Zhang and
               Jieqi Peng and
               Shuai Zhang and
               Jianxin Li and
               Hui Xiong and
               Wancai Zhang},
  title     = {Informer: Beyond Efficient Transformer for Long Sequence Time-Series Forecasting},
  booktitle = {The Thirty-Fifth {AAAI} Conference on Artificial Intelligence, {AAAI} 2021, Virtual Conference},
  volume    = {35},
  number    = {12},
  pages     = {11106--11115},
  publisher = {{AAAI} Press},
  year      = {2021},
}

@misc{silso2015sunspot,
  author       = {Clette, Fr{\'e}d{\'e}ric and Lef{\`e}vre, Laure},
  title        = {{SILSO Sunspot Number Version 2.0}},
  year         = {2015},
  month        = jul,
  publisher    = {World Data Center for the Sunspot Index and Long-term Solar Observations (WDC-SILSO), Royal Observatory of Belgium},
  doi          = {10.24414/qnza-ac80},
  url          = {https://doi.org/10.24414/qnza-ac80},
  note         = {Accessed via WDC-SILSO}
}

@inbook{10.1137/1.9781611972757.40,
author = {Michail Vlachos and Philip Yu and Vittorio Castelli},
title = {On Periodicity Detection and Structural Periodic Similarity},
booktitle = {Proceedings of the 2005 SIAM International Conference on Data Mining (SDM)},
chapter = {},
pages = {449-460},
doi = {10.1137/1.9781611972757.40},
URL = {https://epubs.siam.org/doi/abs/10.1137/1.9781611972757.40},
eprint = {https://epubs.siam.org/doi/pdf/10.1137/1.9781611972757.40}
}

@InProceedings{10.1007/978-3-030-39098-3_4,
author="Puech, Tom
and Boussard, Matthieu
and D'Amato, Anthony
and Millerand, Ga{\"e}tan",
editor="Lemaire, Vincent
and Malinowski, Simon
and Bagnall, Anthony
and Bondu, Alexis
and Guyet, Thomas
and Tavenard, Romain",
title="A Fully Automated Periodicity Detection in Time Series",
booktitle="Advanced Analytics and Learning on Temporal Data",
year="2020",
publisher="Springer International Publishing",
address="Cham",
pages="43--54",
isbn="978-3-030-39098-3"
}

@Article{Toller2019,
author={Toller, Maximilian
and Santos, Tiago
and Kern, Roman},
title={SAZED: parameter-free domain-agnostic season length estimation in time series data},
journal={Data Mining and Knowledge Discovery},
year={2019},
month={Nov},
day={01},
volume={33},
number={6},
pages={1775-1798},
issn={1573-756X},
doi={10.1007/s10618-019-00645-z},
url={https://doi.org/10.1007/s10618-019-00645-z}
}

@article{ljung1978measure,
  title={On a measure of lack of fit in time series models},
  author={Ljung, Greta M and Box, George EP},
  journal={Biometrika},
  volume={65},
  number={2},
  pages={297--303},
  year={1978},
  publisher={Oxford University Press}
}
% \input{ms.bbl}

\clearpage
\appendix

\section{Supplementary Materials}
\label{app:supplementary}

We provide additional technical details to support the findings presented in the main text. This document includes a comprehensive sensitivity analysis of the \textsc{LGTD} hyperparameters, the full experimental protocol used for benchmarking, and an extended set of quantitative and visual results across all evaluated datasets. These materials are intended to facilitate the reproduction of our experiments and provide a deeper audit of the model's performance under varied non-stationary conditions.

\subsection{Sensitivity Analysis}
\label{app:sensitivity}

To assess the robustness of \textsc{LGTD}, we conducted a sensitivity analysis on its two primary internal hyperparameters: the local window size $W$, which governs the locality of the trend-seasonal separation, and the error percentile $p$ used within the AutoTrend-LLT module. We evaluated the decomposition Mean Squared Error (MSE) over a dense grid of discrete parameter values while keeping all other components constant.

Figure~\ref{fig:sensitivity} visualizes the error surfaces across a suite of synthetic datasets featuring varied periodic structures and non-linear trends. The results demonstrate that \textsc{LGTD} is well-conditioned; the error surfaces exhibit smooth gradients without sharp bifurcations or isolated failure regions. Notably, while the optimal $(W, p)$ configuration may shift slightly depending on the dataset's signal-to-noise ratio, the default setting $(W=5, p=50)$ consistently resides within a broad, low-error basin. This indicates that \textsc{LGTD} can achieve near-optimal performance across diverse scenarios without the need for exhaustive per-dataset tuning.

\begin{figure*}[t]
    \centering
    \includegraphics[width=\linewidth]{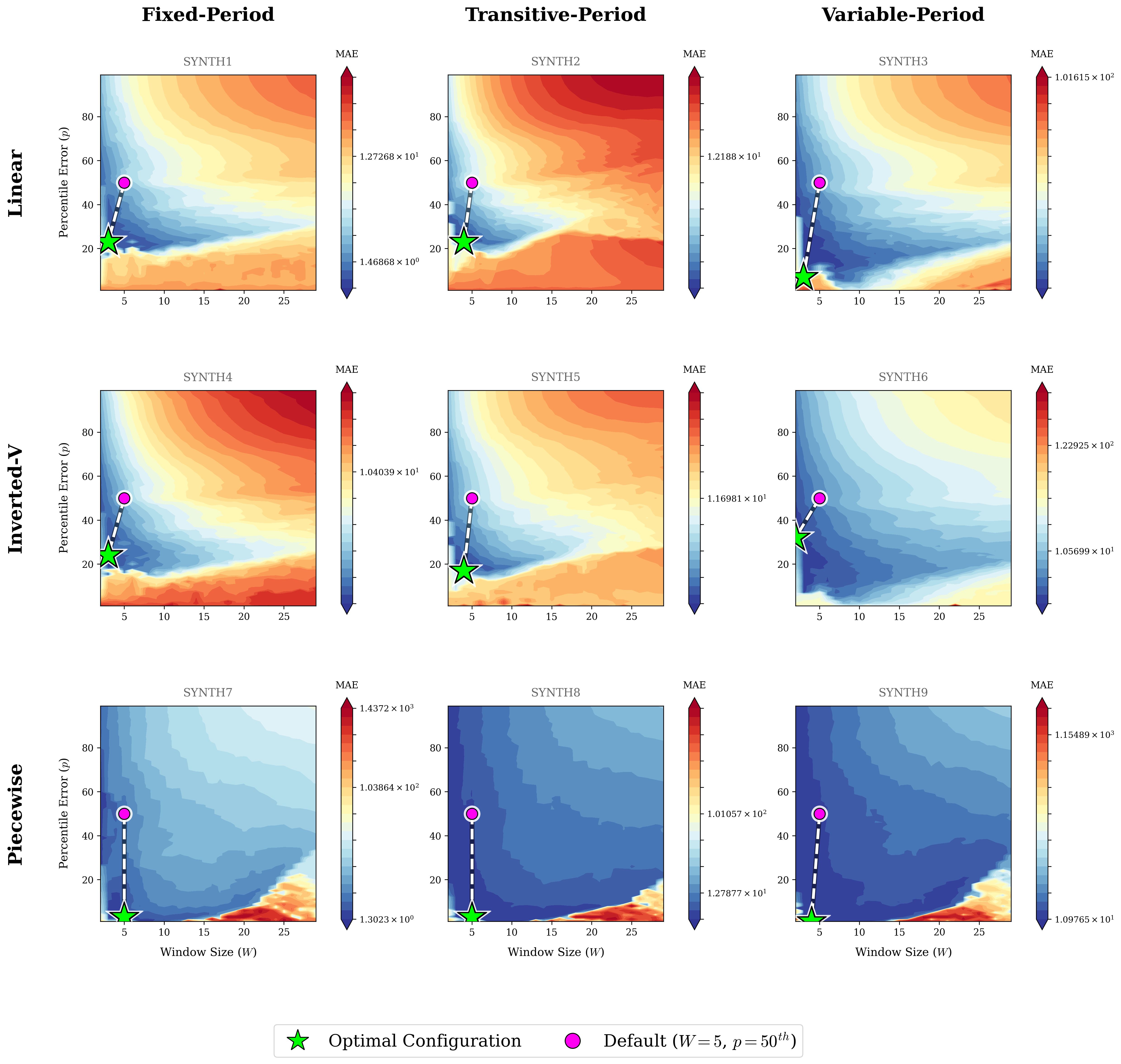}
    \caption{Sensitivity analysis of \textsc{LGTD} with respect to window size $W$ and error percentile $p$. The heat map reports decomposition error across configurations. Performance varies smoothly, indicating that the model is well-conditioned with respect to hyperparameter choice.}
    \label{fig:sensitivity}
\end{figure*}

\subsection{Hyperparameter Configurations}
\label{app:full_params}

For the sake of transparency and reproducibility, we document the complete hyperparameter settings for all baseline methods and \textsc{LGTD} in Table~\ref{tab:model_parameters}. Parameters were sourced from the original authors' recommended settings or determined via standard grid search where documentation was insufficient.

\begin{table*}[!htb]
\centering
\setlength{\tabcolsep}{2pt}
\renewcommand{\arraystretch}{1.05}
\caption{Hyperparameter configurations for all models across synthetic and real-world datasets.}
\label{tab:model_parameters}
\begin{tabular}{llcccccccccccc}
\toprule
& & \multicolumn{9}{c}{\textbf{Synthetic Datasets}} & \multicolumn{3}{c}{\textbf{Real-World Datasets}} \\
\cmidrule(lr){3-11} \cmidrule(lr){12-14}
\textbf{Model} & \textbf{Parameter} & \textbf{s1} & \textbf{s2} & \textbf{s3} & \textbf{s4} & \textbf{s5} & \textbf{s6} & \textbf{s7} & \textbf{s8} & \textbf{s9} & \textbf{ETTh1} & \textbf{ETTh2} & \textbf{Sunspot} \\
\midrule
\multirow{2}{*}{LGTD}
 & Window Size & \multicolumn{12}{c}{\underline{5}} \\
 & Error Percentile & \multicolumn{12}{c}{\underline{50}} \\
\midrule
\multirow{1}{*}{ASTD}
 & Seasonality Smoothing & \multicolumn{12}{c}{\underline{0.7}} \\
\midrule
\multirow{2}{*}{ASTD$_{\text{Online}}$}
 & Seasonality Smoothing & \multicolumn{12}{c}{\underline{0.7}} \\
 & Init Window Size & \multicolumn{12}{c}{\underline{300}} \\
\midrule
\multirow{10}{*}{FastRobustSTL}
 & Period & \underline{120} & \underline{120} & \underline{120} & \underline{120} & 60 & \underline{120} & \underline{120} & \underline{120} & \underline{120} & 24 & 24 & 12 \\
 & $\lambda_1$ (Trend) & 1 & 1 & 1 & 1 & 5 & 1 & 1 & 1 & 1 & 1 & 1 & 1 \\
 & $\lambda_2$ (Seasonal) & 10.00 & 10.00 & 10.00 & 10.00 & 20.00 & 10.00 & 10.00 & 10.00 & 10.00 & 10.00 & 10.00 & 10.00 \\
 & $K$ (Bilateral) & \underline{2} & \underline{2} & \underline{2} & \underline{2} & 1 & \underline{2} & \underline{2} & \underline{2} & \underline{2} & \underline{2} & \underline{2} & 1 \\
 & $H$ (Trend) & \underline{5} & \underline{5} & \underline{5} & \underline{5} & 3 & \underline{5} & \underline{5} & \underline{5} & \underline{5} & \underline{5} & \underline{5} & 2 \\
 & $d_{n1}$ & \underline{1} & \underline{1} & \underline{1} & \underline{1} & 1.5 & \underline{1} & \underline{1} & \underline{1} & \underline{1} & \underline{1} & \underline{1} & 5 \\
 & $d_{n2}$ & \underline{1} & \underline{1} & \underline{1} & \underline{1} & 1.5 & \underline{1} & \underline{1} & \underline{1} & \underline{1} & \underline{1} & \underline{1} & 5 \\
 & $d_{s1}$ & \underline{50.00} & \underline{50.00} & \underline{50.00} & \underline{50.00} & 75.00 & \underline{50.00} & \underline{50.00} & \underline{50.00} & \underline{50.00} & \underline{50.00} & \underline{50.00} & 100.00 \\
 & $d_{s2}$ & \underline{1} & \underline{1} & \underline{1} & \underline{1} & 1.5 & \underline{1} & \underline{1} & \underline{1} & \underline{1} & \underline{1} & \underline{1} & 5 \\
 & Max Iterations & \underline{1000} & \underline{1000} & \underline{1000} & \underline{1000} & 800 & \underline{1000} & \underline{1000} & \underline{1000} & \underline{1000} & \underline{1000} & \underline{1000} & 20 \\
\midrule
\multirow{3}{*}{OneShotSTL}
 & Period & \underline{120} & 60 & \underline{120} & \underline{120} & \underline{120} & \underline{120} & \underline{120} & \underline{120} & \underline{120} & 24 & 24 & 12 \\
 & Init Ratio & \multicolumn{12}{c}{\underline{0.3}} \\
 & Shift Window & \multicolumn{12}{c}{\underline{0}} \\
\midrule
\multirow{3}{*}{OnlineSTL}
 & Periods & \underline{[120]} & [60] & \underline{[120]} & \underline{[120]} & \underline{[120]} & \underline{[120]} & \underline{[120]} & \underline{[120]} & \underline{[120]} & [24] & [24] & [12] \\
 & $\lambda$ (Smoothing) & 0.3 & 0.1 & 0.3 & 0.3 & 0.3 & 0.3 & 0.3 & 0.3 & 0.3 & 0.3 & 0.3 & 0.3 \\
 & Init Window Ratio & \multicolumn{12}{c}{\underline{0.3}} \\
\midrule
\multirow{11}{*}{STL}
 & Period & \underline{120} & 60 & \underline{120} & \underline{120} & \underline{120} & \underline{120} & \underline{120} & \underline{120} & \underline{120} & 24 & 24 & 12 \\
 & Seasonal Window & \multicolumn{12}{c}{\underline{13}} \\
 & Trend Window & \multicolumn{12}{c}{} \\
 & Robust & \multicolumn{12}{c}{\underline{$\times$}} \\
 & Low Pass & \multicolumn{12}{c}{} \\
 & Low Pass Deg & \multicolumn{12}{c}{1} \\
 & Low Pass Jump & \multicolumn{12}{c}{1} \\
 & Seasonal Deg & \multicolumn{12}{c}{1} \\
 & Seasonal Jump & \multicolumn{12}{c}{1} \\
 & Trend Deg & \multicolumn{12}{c}{1} \\
 & Trend Jump & \multicolumn{12}{c}{1} \\
\midrule
\multirow{6}{*}{STR}
 & Seasonal Periods & \underline{[120]} & \underline{[120]} & \underline{[120]} & \underline{[120]} & \underline{[120]} & \underline{[120]} & \underline{[120]} & \underline{[120]} & \underline{[120]} & [24] & [24] & [12] \\
 & Trend $\lambda$ & \multicolumn{12}{c}{\underline{1000.00}} \\
 & Seasonal $\lambda$ & \multicolumn{12}{c}{100.00} \\
 & Robust & \multicolumn{12}{c}{\underline{$\times$}} \\
 & Auto Params & \multicolumn{12}{c}{\underline{$\times$}} \\
 & N Trials & \multicolumn{12}{c}{\underline{10}} \\
\bottomrule
\end{tabular}
\end{table*}

\subsection{Complete Decomposition Results}
\label{app:full_results}

We provide a granular breakdown of decomposition errors (MSE/MAE) for each structural component: Trend ($T$), Seasonal ($S$), and Residual ($R$). Table~\ref{tab:decomposition_transposed} presents these metrics across three distinct periodicity regimes: \textbf{Fixed}, \textbf{Transitive}, and \textbf{Variable}. 

This detailed reporting reveals that \textsc{LGTD}'s primary advantage lies in its ability to suppress cross-component leakage. In the "Variable Period" regime, traditional methods often erroneously attribute frequency shifts to the trend or residual components, whereas \textsc{LGTD} maintains a sparse residual and a coherent seasonal signal..

\newcolumntype{Y}{>{\centering\arraybackslash}X}

\begin{table*}[!p]
\centering
\footnotesize
\setlength{\tabcolsep}{2pt}
\renewcommand{\arraystretch}{1.0}
\caption{Decomposition errors (MSE/MAE) across synthetic datasets.}
\label{tab:decomposition_transposed}
\begin{tabularx}{\textwidth}{llllYYYYYYYY}
\toprule
\textbf{Trend} & \textbf{Period} & \textbf{Comp.} & \textbf{Metric} & LGTD & STL & STR & FastRobustSTL & $ASTD$ & $ASTD_{Online}$ & OnlineSTL & OneShotSTL \\
\midrule
\multirow{18}{*}{\rotatebox{90}{\textbf{Linear}}} & \multirow{6}{*}{\textbf{Fixed}} & \multirow{2}{*}{Trend} & MSE & 0.15 & \textbf{0.01} & 4.06 & \underline{0.08} & 8.06 & 7.81 & 5.35 & 18.82 \\
 &  &  & MAE & 0.33 & \textbf{0.08} & 1.50 & \underline{0.22} & 2.37 & 2.10 & 1.53 & 3.45 \\ \addlinespace[2pt]
 &  & \multirow{2}{*}{Seasonal} & MSE & 14.65 & \textbf{0.16} & 164.41 & \underline{1.29} & 5.21 & 61.89 & 377.12 & 377.50 \\
 &  &  & MAE & 3.12 & \textbf{0.32} & 10.88 & \underline{0.88} & 1.85 & 6.94 & 10.55 & 10.46 \\ \addlinespace[2pt]
 &  & \multirow{2}{*}{Residual} & MSE & 14.48 & \textbf{0.17} & 128.30 & 1.29 & 3.87 & 61.70 & 6.18 & \underline{0.60} \\
 &  &  & MAE & 3.10 & \textbf{0.32} & 9.45 & 0.89 & 1.59 & 6.98 & 2.10 & \underline{0.58} \\
 \cmidrule(lr){2-12}
 & \multirow{6}{*}{\textbf{Transitive}} & \multirow{2}{*}{Trend} & MSE & \textbf{3.00} & \underline{4.87} & 41.85 & 852.32 & 7.80 & 10.42 & 5.99 & 2220.42 \\
 &  &  & MAE & \underline{1.50} & \textbf{1.21} & 4.13 & 25.36 & 2.32 & 2.49 & 1.79 & 33.17 \\ \addlinespace[2pt]
 &  & \multirow{2}{*}{Seasonal} & MSE & \textbf{32.17} & 216.54 & 1196.24 & 861.03 & \underline{46.68} & 96.14 & 383.42 & 2531.68 \\
 &  &  & MAE & \textbf{4.80} & 10.27 & 31.10 & 25.49 & \underline{5.04} & 7.93 & 11.31 & 41.18 \\ \addlinespace[2pt]
 &  & \multirow{2}{*}{Residual} & MSE & 28.68 & 213.14 & 959.29 & \textbf{1.38} & 45.79 & 95.99 & \underline{12.06} & 30.85 \\
 &  &  & MAE & 4.59 & 10.27 & 27.66 & \textbf{0.87} & 5.05 & 7.95 & \underline{2.80} & 3.72 \\
 \cmidrule(lr){2-12}
 & \multirow{6}{*}{\textbf{Variable}} & \multirow{2}{*}{Trend} & MSE & \textbf{0.70} & 169.32 & 203.94 & 1204.47 & \underline{8.10} & 25.50 & 52.66 & 934.91 \\
 &  &  & MAE & \textbf{0.72} & 8.79 & 10.14 & 30.76 & \underline{2.42} & 3.88 & 6.29 & 22.36 \\ \addlinespace[2pt]
 &  & \multirow{2}{*}{Seasonal} & MSE & \textbf{22.93} & 1190.37 & 1230.83 & 1213.92 & \underline{146.11} & 174.49 & 521.91 & 1453.29 \\
 &  &  & MAE & \textbf{3.92} & 30.59 & 31.48 & 30.89 & \underline{9.16} & 11.30 & 15.27 & 31.05 \\ \addlinespace[2pt]
 &  & \multirow{2}{*}{Residual} & MSE & \underline{23.00} & 762.63 & 650.54 & \textbf{1.51} & 143.36 & 185.76 & 82.20 & 300.92 \\
 &  &  & MAE & \underline{3.96} & 23.42 & 21.86 & \textbf{0.95} & 9.17 & 11.60 & 7.82 & 14.29 \\
\midrule
\multirow{18}{*}{\rotatebox{90}{\textbf{Inverted-V}}} & \multirow{6}{*}{\textbf{Fixed}} & \multirow{2}{*}{Trend} & MSE & 5.14 & \underline{1.26} & 2.47 & \textbf{0.69} & 1425.59 & 162.75 & 196.99 & 27.27 \\
 &  &  & MAE & 1.02 & \textbf{0.45} & 1.40 & \underline{0.60} & 28.92 & 8.64 & 10.83 & 4.14 \\ \addlinespace[2pt]
 &  & \multirow{2}{*}{Seasonal} & MSE & 33.77 & \textbf{0.23} & 2.95 & \underline{2.04} & 884.76 & 52.67 & 500.45 & 377.96 \\
 &  &  & MAE & 4.97 & \textbf{0.38} & 1.49 & \underline{1.08} & 23.98 & 5.38 & 16.39 & 10.68 \\ \addlinespace[2pt]
 &  & \multirow{2}{*}{Residual} & MSE & 27.37 & \underline{1.33} & \textbf{0.13} & 1.47 & 134.15 & 65.65 & 56.74 & 2.62 \\
 &  &  & MAE & 4.60 & \underline{0.62} & \textbf{0.26} & 0.93 & 9.28 & 5.99 & 5.78 & 1.32 \\
 \cmidrule(lr){2-12}
 & \multirow{6}{*}{\textbf{Transitive}} & \multirow{2}{*}{Trend} & MSE & \underline{8.52} & \textbf{3.52} & 354.59 & 248.04 & 1425.70 & 184.69 & 209.36 & 2695.13 \\
 &  &  & MAE & \underline{1.42} & \textbf{1.04} & 15.68 & 8.63 & 28.93 & 9.18 & 11.16 & 38.17 \\ \addlinespace[2pt]
 &  & \multirow{2}{*}{Seasonal} & MSE & \textbf{36.19} & 215.46 & 991.10 & 308.80 & 949.09 & \underline{104.47} & 502.50 & 2787.26 \\
 &  &  & MAE & \textbf{4.98} & 10.26 & 27.44 & 9.68 & 24.04 & \underline{7.10} & 16.47 & 43.59 \\ \addlinespace[2pt]
 &  & \multirow{2}{*}{Residual} & MSE & \textbf{27.00} & 213.52 & 231.39 & \underline{30.22} & 191.02 & 123.40 & 52.99 & 32.82 \\
 &  &  & MAE & 4.46 & 10.22 & 13.27 & \textbf{2.24} & 10.63 & 8.02 & 5.67 & \underline{3.72} \\
 \cmidrule(lr){2-12}
 & \multirow{6}{*}{\textbf{Variable}} & \multirow{2}{*}{Trend} & MSE & \underline{165.15} & 167.07 & 201.12 & 1203.79 & 1478.09 & \textbf{148.42} & 246.83 & 984.61 \\
 &  &  & MAE & 8.79 & \underline{8.65} & 9.99 & 30.72 & 29.51 & \textbf{8.52} & 13.30 & 24.05 \\ \addlinespace[2pt]
 &  & \multirow{2}{*}{Seasonal} & MSE & \underline{205.34} & 1190.60 & 1230.80 & 1213.80 & 840.08 & \textbf{195.40} & 643.69 & 1440.47 \\
 &  &  & MAE & \textbf{9.91} & 30.59 & 31.48 & 30.85 & 22.95 & \underline{11.77} & 17.43 & 31.17 \\ \addlinespace[2pt]
 &  & \multirow{2}{*}{Residual} & MSE & \underline{31.84} & 775.36 & 659.87 & \textbf{1.53} & 282.29 & 235.45 & 124.94 & 311.42 \\
 &  &  & MAE & \underline{4.92} & 23.60 & 22.01 & \textbf{0.96} & 11.91 & 12.73 & 9.60 & 14.53 \\
\midrule
\multirow{18}{*}{\rotatebox{90}{\textbf{Piecewise}}} & \multirow{6}{*}{\textbf{Fixed}} & \multirow{2}{*}{Trend} & MSE & 36.34 & \underline{0.67} & 4.00 & \textbf{0.55} & 392.57 & 117.26 & 372.37 & 334.28 \\
 &  &  & MAE & 2.93 & \textbf{0.36} & 1.79 & \underline{0.57} & 16.92 & 8.48 & 14.53 & 13.01 \\ \addlinespace[2pt]
 &  & \multirow{2}{*}{Seasonal} & MSE & 54.04 & \textbf{0.18} & 4.63 & \underline{1.91} & 294.82 & 166.10 & 1031.33 & 967.00 \\
 &  &  & MAE & 4.71 & \textbf{0.33} & 1.90 & \underline{1.08} & 14.79 & 11.23 & 20.94 & 16.97 \\ \addlinespace[2pt]
 &  & \multirow{2}{*}{Residual} & MSE & 15.93 & \underline{0.80} & \textbf{0.21} & 1.63 & 35.32 & 188.70 & 26.99 & 2.07 \\
 &  &  & MAE & 3.30 & \underline{0.57} & \textbf{0.32} & 0.99 & 4.92 & 11.64 & 4.16 & 1.16 \\
 \cmidrule(lr){2-12}
 & \multirow{6}{*}{\textbf{Transitive}} & \multirow{2}{*}{Trend} & MSE & \underline{66.37} & \textbf{13.12} & 1826.80 & 2263.08 & 394.80 & 142.17 & 390.21 & 5866.60 \\
 &  &  & MAE & \underline{3.45} & \textbf{2.09} & 37.00 & 41.34 & 16.99 & 9.24 & 15.09 & 56.48 \\ \addlinespace[2pt]
 &  & \multirow{2}{*}{Seasonal} & MSE & \textbf{92.60} & 554.08 & 2491.78 & 2281.51 & 404.94 & \underline{258.81} & 1057.93 & 6396.83 \\
 &  &  & MAE & \textbf{6.17} & 16.42 & 43.46 & 41.51 & 16.85 & \underline{13.21} & 22.01 & 64.55 \\ \addlinespace[2pt]
 &  & \multirow{2}{*}{Residual} & MSE & \underline{25.70} & 545.83 & 154.09 & \textbf{2.21} & 144.66 & 272.68 & 42.46 & 44.86 \\
 &  &  & MAE & \underline{3.93} & 16.44 & 8.89 & \textbf{1.00} & 9.28 & 13.43 & 4.95 & 4.31 \\
 \cmidrule(lr){2-12}
 & \multirow{6}{*}{\textbf{Variable}} & \multirow{2}{*}{Trend} & MSE & \underline{337.21} & 432.92 & 2773.16 & 2927.05 & 395.73 & \textbf{187.87} & 623.25 & 4390.01 \\
 &  &  & MAE & 15.51 & \underline{14.08} & 46.92 & 47.34 & 16.93 & \textbf{10.95} & 19.12 & 51.85 \\ \addlinespace[2pt]
 &  & \multirow{2}{*}{Seasonal} & MSE & \textbf{360.96} & 3046.73 & 3120.35 & 2999.48 & 649.97 & \underline{462.51} & 1438.61 & 4612.60 \\
 &  &  & MAE & \textbf{16.05} & 48.94 & 49.82 & 47.80 & 20.13 & \underline{18.15} & 26.11 & 56.39 \\ \addlinespace[2pt]
 &  & \multirow{2}{*}{Residual} & MSE & \underline{25.45} & 1955.67 & 34.74 & \textbf{5.92} & 391.35 & 500.48 & 184.18 & 501.44 \\
 &  &  & MAE & 4.22 & 37.52 & \underline{3.46} & \textbf{1.54} & 15.23 & 18.68 & 11.82 & 16.35 \\
\bottomrule
\end{tabularx}
\end{table*}

\subsection{Visual Decomposition Examples}
\label{app:visual_comparisons}

Figure~\ref{fig:synth_all} presents decomposition results for representative
synthetic datasets, organized by trend type (rows) and seasonality pattern
(columns). Each block shows the original series with the estimated global trend,
the extracted seasonal component, and the corresponding residual. The datasets
span linear, inverted-V, and piecewise trends under fixed, transitive, and
variable season-length regimes, providing a controlled comparison across
increasingly nonstationary settings.

Across all scenarios, \textsc{LGTD} consistently recovers smooth and coherent
global trends while adapting the seasonal component to changes in cycle length
and amplitude. In contrast, period-based methods exhibit visible distortion in
the seasonal component and increased residual structure when seasonality drifts
or transitions. These visual results complement the quantitative MAE evaluations
in Section~\ref{sec:results}, illustrating how \textsc{LGTD} suppresses
cross-component leakage and maintains low-variance residuals without requiring
explicit season-length specification.

\begin{figure*}[!htb]
  \centering
  \small
  \begin{tabular}{@{}cccc@{}}
    & \textbf{Fixed-Period} & \textbf{Transitive-Period} & \textbf{Variable-Period} \\
    \rotatebox{90}{\hspace{1em}\textbf{Linear}} & 
    \includegraphics[width=0.3\textwidth]{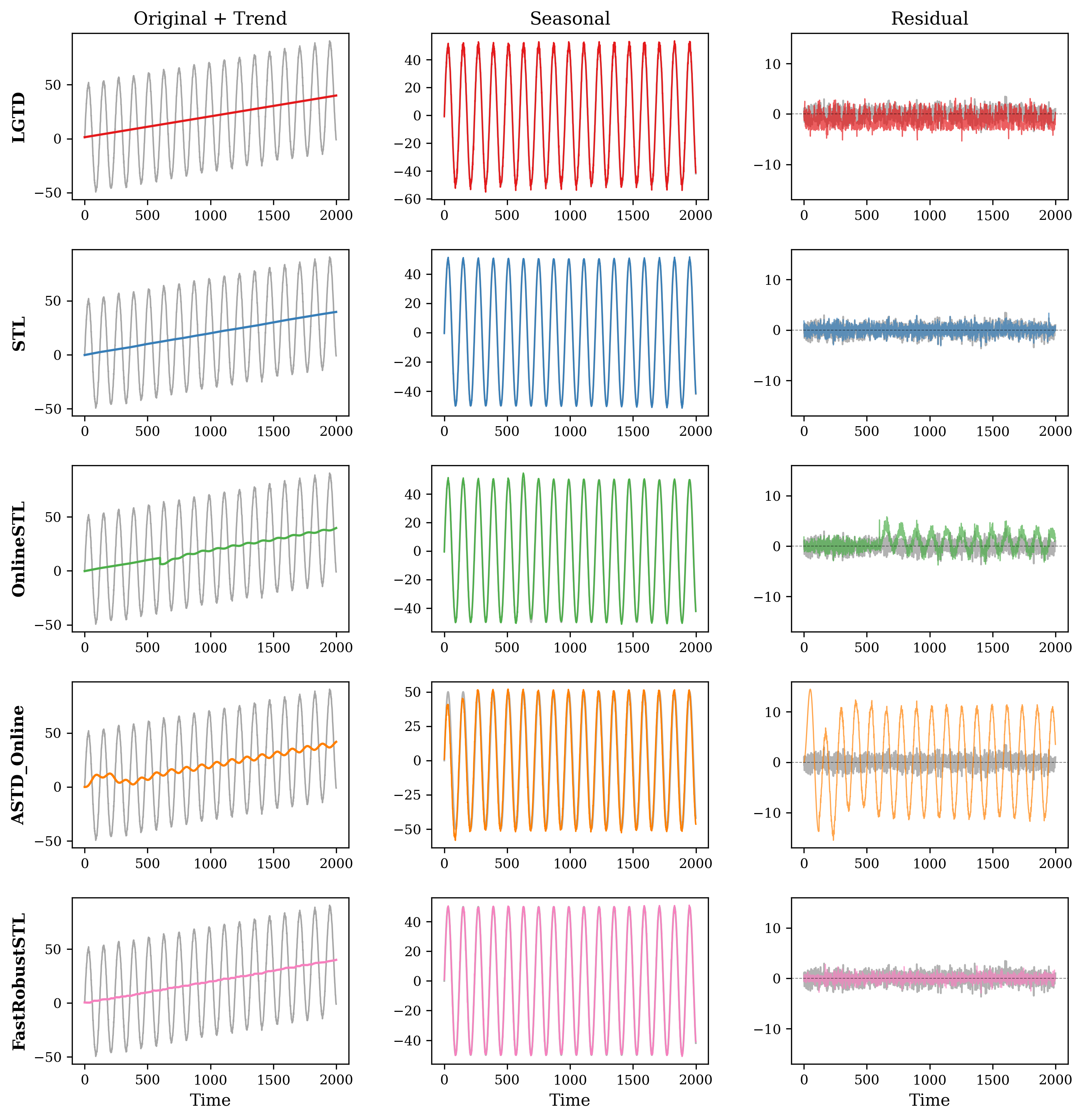} &
    \includegraphics[width=0.3\textwidth]{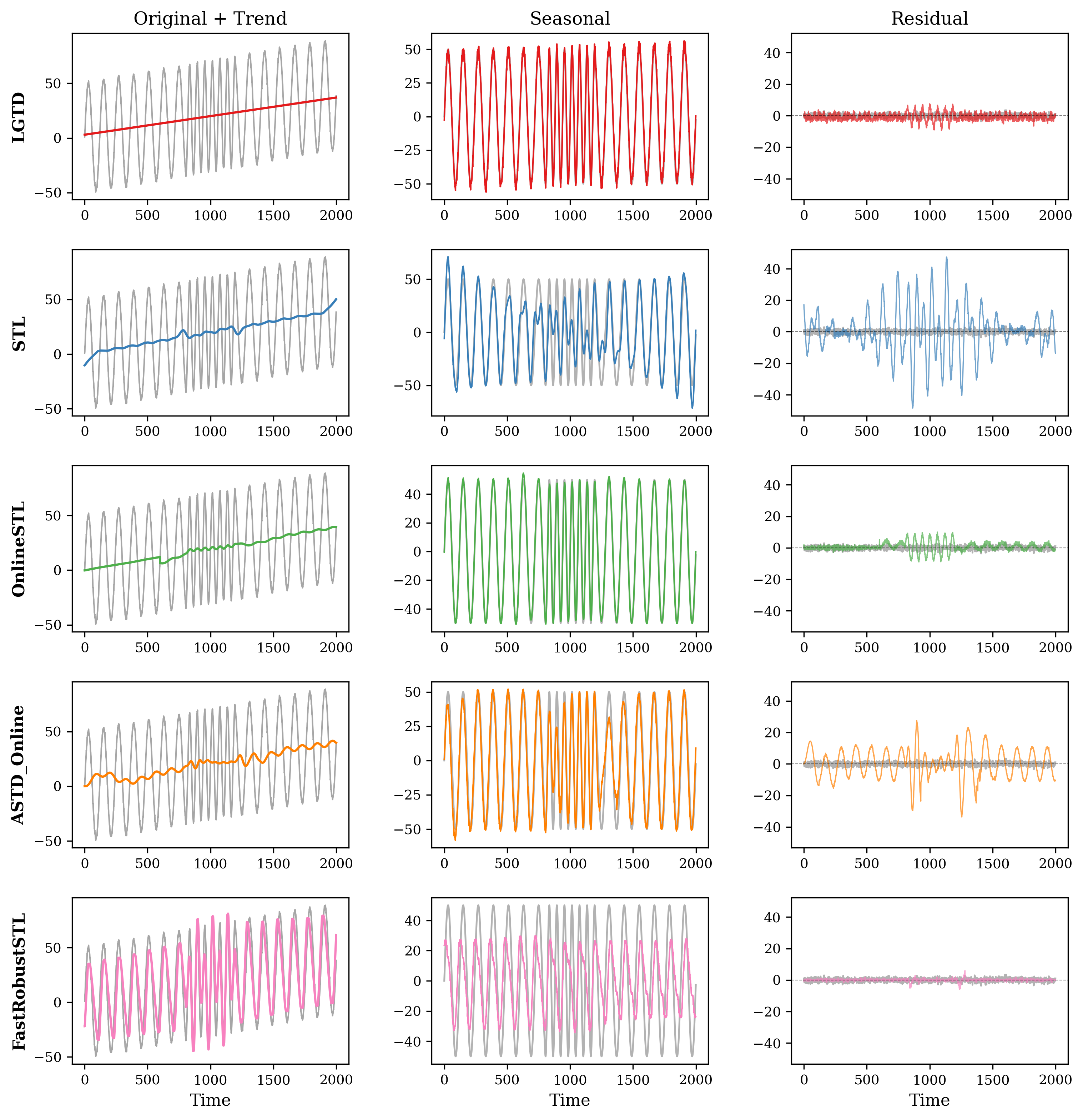} &
    \includegraphics[width=0.3\textwidth]{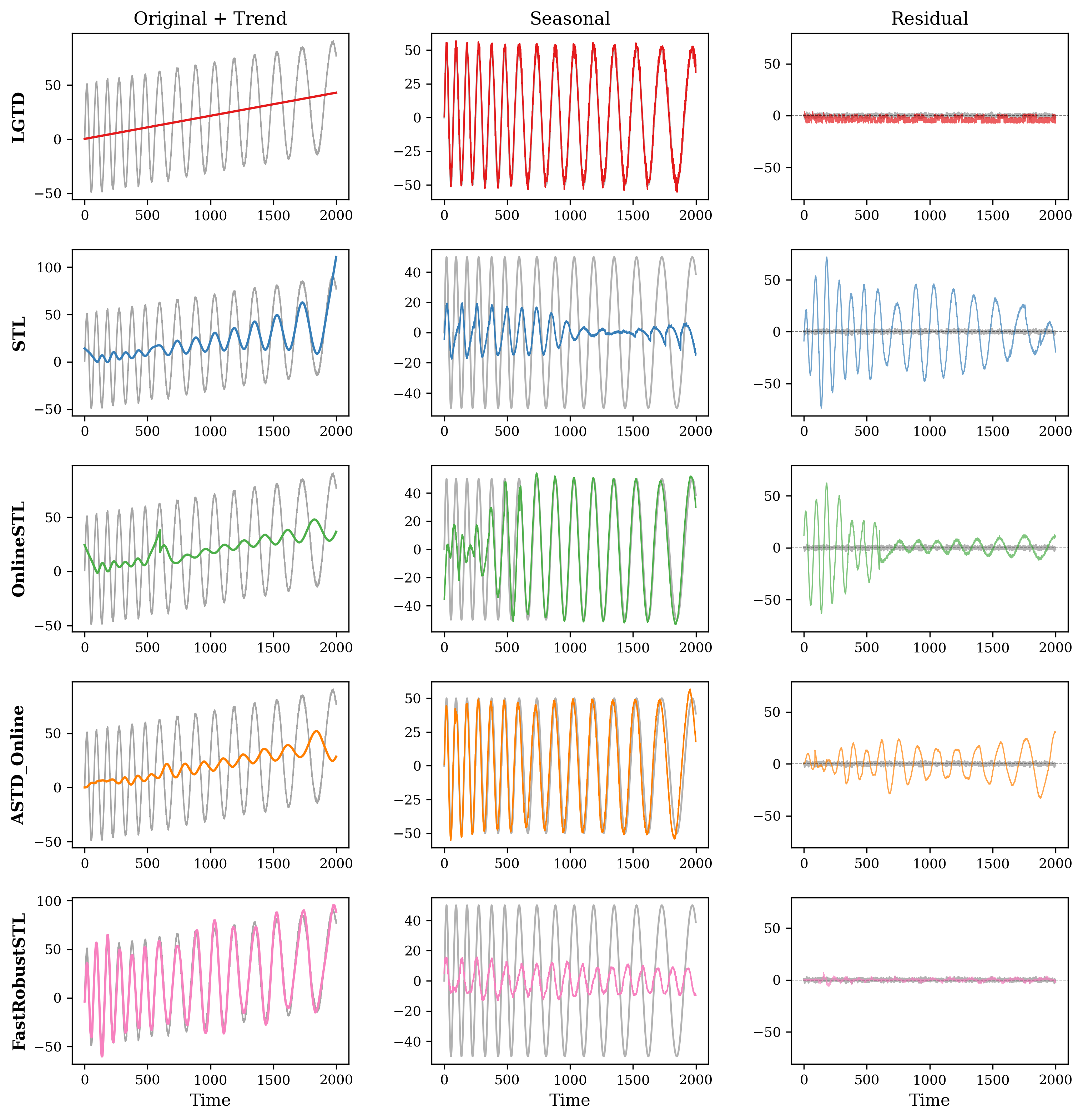} \\
    & (a) synth1 & (b) synth2 & (c) synth3 \\[0.5em]
    
    \rotatebox{90}{\hspace{1em}\textbf{Inverted-V}} & 
    \includegraphics[width=0.3\textwidth]{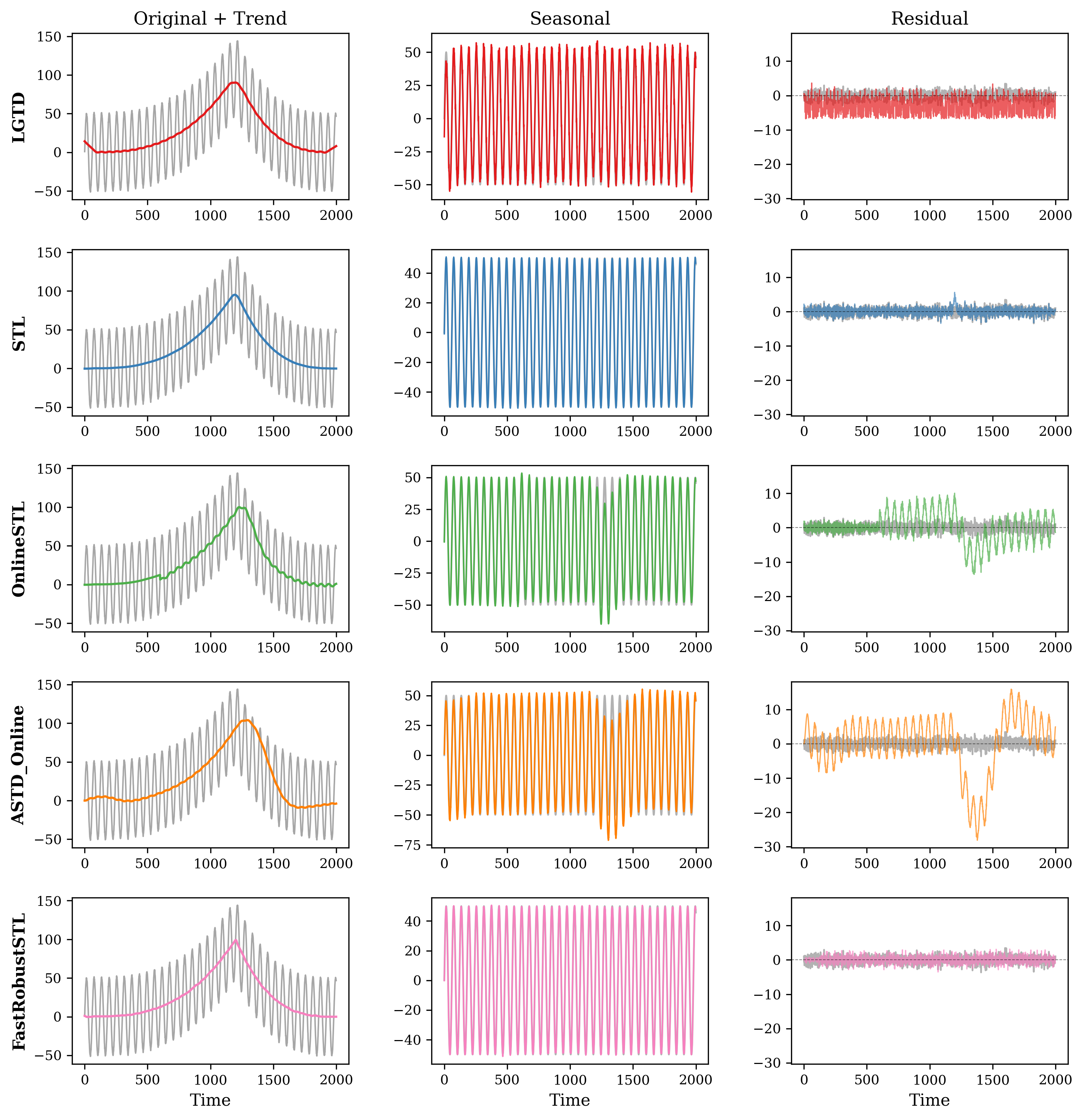} &
    \includegraphics[width=0.3\textwidth]{FIG/SYNTHETIC/synth5_comparison.png} &
    \includegraphics[width=0.3\textwidth]{FIG/SYNTHETIC/synth8_comparison.png} \\
    & (d) synth4 & (e) synth5 & (f) synth6 \\[0.5em]
    
    \rotatebox{90}{\hspace{1.5em}\textbf{Piecewise}} & 
    \includegraphics[width=0.3\textwidth]{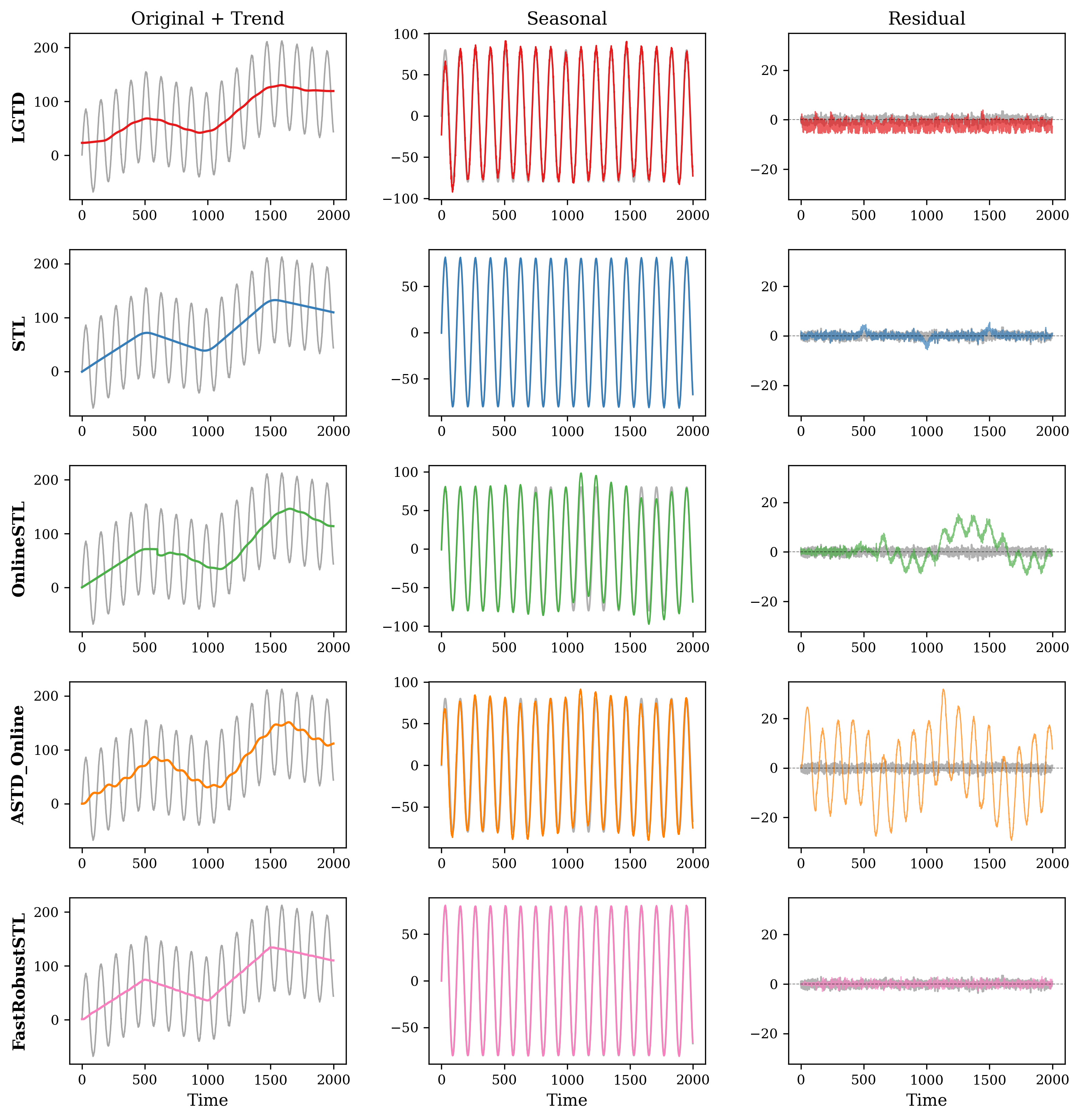} &
    \includegraphics[width=0.3\textwidth]{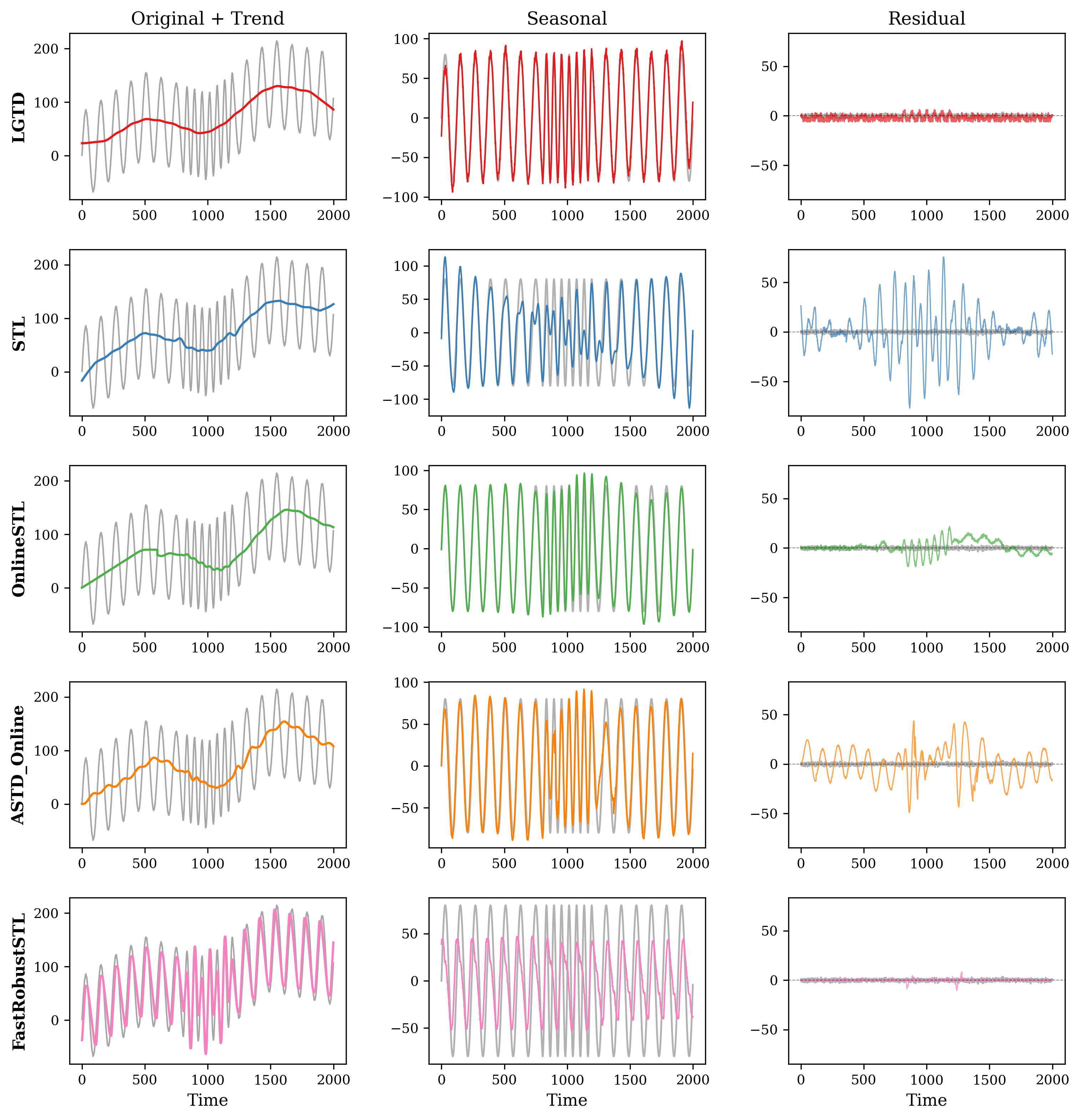} &
    \includegraphics[width=0.3\textwidth]{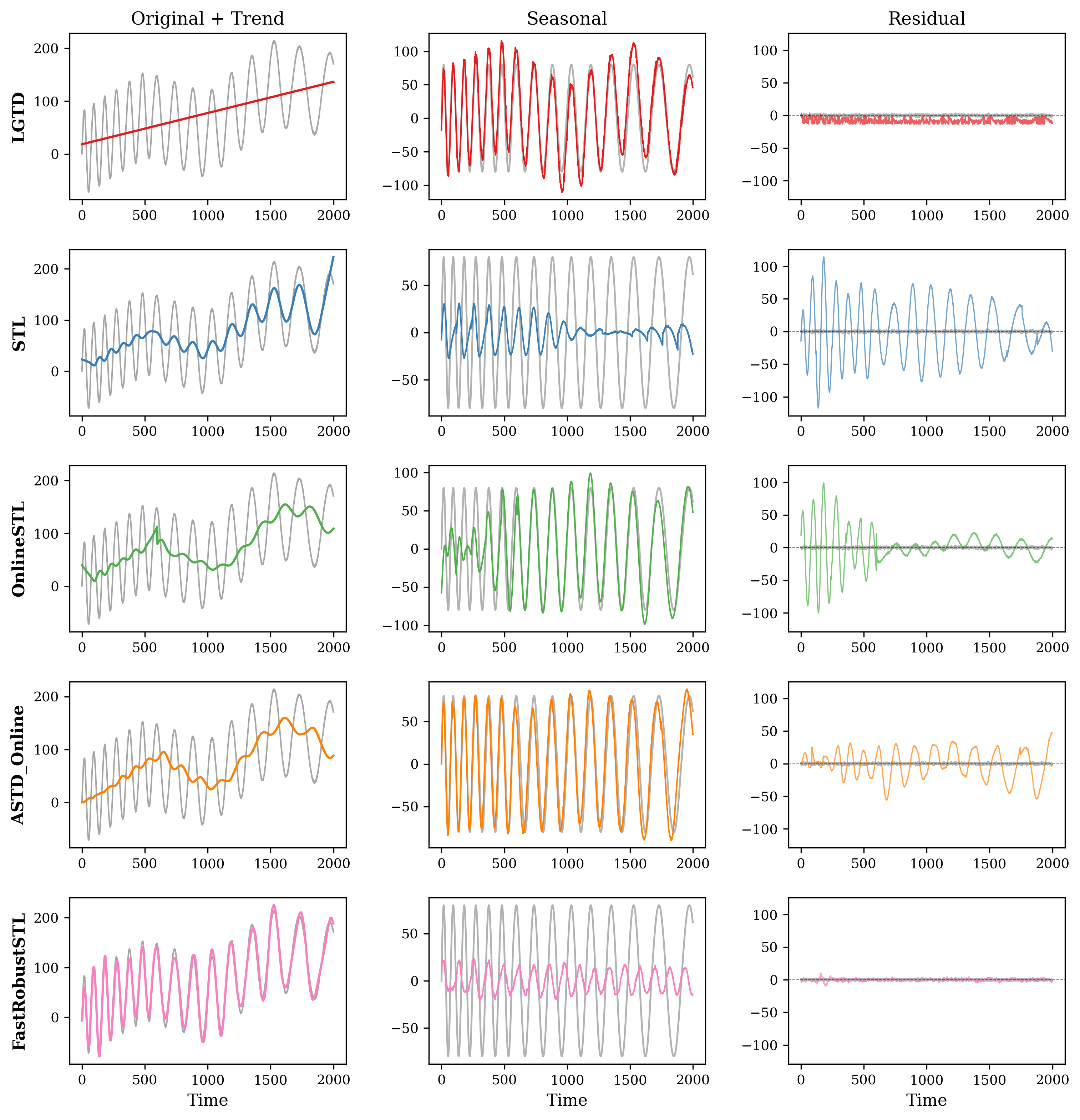} \\
    & (g) synth7 & (h) synth8 & (i) synth9
  \end{tabular}
  \caption{Decomposition comparisons across synthetic datasets. Rows represent trend types; columns represent periodicity variations.}
  \label{fig:synth_all}
\end{figure*}

\subsection{Source Code and Data Availability}
\label{app:code}
To support the reproducibility of this research, we provide the full implementation and datasets at the following anonymous repository: \url{https://github.com/chotanansub/LGTD}. The repository includes a configuration file to recreate the exact environment used for our experiments.








\end{document}